\documentclass[journal,onecolumn]{IEEEtran}
\usepackage{cite}

\usepackage{perlaza}
\usepackage{tikz}
\usepackage{amsmath}
\usepackage[hidelinks]{hyperref}
\usepackage{url}
\usepackage{pifont}
\usepackage{longtable}
\usepackage{tabularx}
\usepackage{makecell, cellspace, caption}
\usepackage{array}
\newcolumntype{L}[1]{>{\raggedright\let\newline\\\arraybackslash\hspace{0pt}}m{#1}}
\newcolumntype{C}[1]{>{\centering\let\newline\\\arraybackslash\hspace{0pt}}m{#1}}
\newcolumntype{R}[1]{>{\raggedleft\let\newline\\\arraybackslash\hspace{0pt}}m{#1}}

\makeatletter
\DeclareFontFamily{U}{tipa}{}
\DeclareFontShape{U}{tipa}{m}{n}{<->tipa10}{}
\newcommand{\arc@char}{{\usefont{U}{tipa}{m}{n}\symbol{62}}}%

\newcommand{\arc}[1]{\mathpalette\arc@arc{#1}}

\newcommand{\arc@arc}[2]{%
  \sbox0{$\m@th#1#2$}%
  \vbox{
    \hbox{\resizebox{\wd0}{\height}{\arc@char}}
    \nointerlineskip
    \box0
  }%
}
\makeatother
\newcommand{\ts}{\textsuperscript}

\DeclareMathOperator{\sinc}{sinc}

\begin{document}
\title{On the Trade-offs Between Information and Energy Transmission in the Finite Block-Length Regime with Finite Channel Inputs}

\author{Sadaf ul Zuhra,~\IEEEmembership{Member,~IEEE,}
Samir M. Perlaza,~\IEEEmembership{Senior Member,~IEEE,}
H. Vincent Poor,~\IEEEmembership{Fellow,~IEEE,}
Mikael Skoglund,~\IEEEmembership{Fellow,~IEEE}
\thanks{Sadaf ul Zuhra, Samir M. Perlaza, and H. Vincent Poor are with the Department of Electrical and Computer Engineering, Princeton University, 08540 Princeton, NJ, USA. $\lbrace$sadaf.zuhra, poor$\rbrace$@princeton.edu\newline
Samir M. Perlaza is with INRIA, Centre Inria d'Universit\'{e} C\^{o}te d'Azur, 2004  Route des Lucioles, 06902 Sophia Antipolis, France. $\lbrace$samir.perlaza$\rbrace$@inria.fr\newline
Mikael Skoglund is with the School of Electrical Engineering and Computer Science, Malvinas V\"{a}g 10,  KTH Royal Institute of Technology, 11428 Stockholm, Sweden. (skoglund@kth.se)\newline
Samir M. Perlaza is also with the Laboratoire de Math\'{e}matiques GAATI, Universit\'{e} de la Polyn\'{e}sie Fran\c{c}aise,  BP 6570, 98702 Faaa, French Polynesia. \newline
%
%
This research was supported in part by the European Commission through the H2020-MSCA-RISE-2019 program under grant 872172; in part by the Agence Nationale de la Recherche (ANR) through the project MAESTRO-5G (ANR-18-CE25-0012); in part by the U.S. National Science Foundation under Grant CCF-1908308; and in part by the French Government through the ``Plan de Relance" and ``Programme d’investissements d’avenir". \newline
This paper was presented in part at the IEEE Information Theory Workshop (ITW) 2021, IEEE International Symposium on Information Theory 2022 and IEEE ITW 2022.
}
}

%

\maketitle

\begin{abstract}
This paper characterizes the trade-offs between information and energy transmission over an additive white Gaussian noise channel in the finite block-length regime with finite sets of channel input symbols.
These trade-offs are characterized using impossibility and achievability bounds on the information transmission rate, energy transmission rate, decoding error probability (DEP) and energy outage probability (EOP) for a finite block-length code.
Given a set of channel input symbols, the impossibility results identify the tuples of information rate, energy rate, DEP and EOP that cannot be achieved by any code using the given set of channel inputs.
A novel method for constructing a family of codes that satisfy a target information rate, energy rate, DEP and EOP is also proposed.
The achievability bounds identify the set of tuples of information rate, energy rate, DEP and EOP that can be simultaneously achieved by the constructed family of codes.
The proposed construction matches the impossibility bounds for the information rate, energy rate, and the EOP.
However, for a given information rate, energy rate and EOP, the achieved DEP  is higher than the impossibility bound due to the choice of the decoding sets made during the code construction.
\end{abstract}

\begin{IEEEkeywords}
Simultaneous information and energy transmission, SIET, SWIPT, information-energy trade-offs, achievability, impossibility bounds, finite block-length, discrete channel inputs, finite channel inputs, finite constellations
\end{IEEEkeywords}

\section{Introduction} \label{sec:intro}
Simultaneous information and energy transmission (SIET) (also known as simultaneous wireless information and power transfer (SWIPT)) employs radio frequency (RF) signals to simultaneously accomplish the tasks of conveying information and providing energy to (possibly different) devices. In the following, for the sake of correctness, the denomination ``energy transmission'' is preferred against ``power transfer'', and thus, the acronym SIET is adopted in the remainder of this paper. 
Given a certain code, a fundamental issue of interest in SIET is, to characterize the fundamental limits on the information rate, energy rate, decoding error probability (DEP) and energy outage probability (EOP) achievable by such a code and to study the trade-offs between these parameters.
The trade-offs between information and energy transmission rates in SIET have previously been studied in the asymptotic regime~\cite{varshney2008transporting,amor2016fundamental,NizarItw2021,GroverSahai,amor2016feedback,KhalfetGIC} where the assumption of infinitely long transmissions guarantees that the DEP and the EOP can be made arbitrarily close to zero. Thus, the focus of the asymptotic results is only on the information and energy transmission rates. In the finite block-length regime, however, the DEP and EOP are bounded away from zero and pose additional constraints on the fundamental limits of SIET. 

A common assumption in the study of SIET over an additive white Gaussian noise (AWGN) channel is that the channel inputs are derived independently from a Gaussian distribution. In a departure from this norm, this work considers the more practical, discrete set of channel input symbols that define the channel. This work identifies the fundamental limits of SIET for such channels in the finite block-length regime in the presence of AWGN. Unlike Gaussian inputs, the fundamental limits obtained for finite discrete sets of channel input symbols provide insights for practical systems that operate with finite block-lengths and finite sets of channel inputs. 

\subsection{State of the art}
The existing body of work in SIET falls into three main categories. The first is the study of the the trade-offs between the information and energy rates that can be simultaneously transmitted by an RF signal. 
Earlier research in this area is focused primarily on the asymptotic regime~\cite{varshney2008transporting,GroverSahai,amor2016fundamental, NizarItw2021}. In this case, the notion of the information-energy region generalizes to the set of all information and energy rate tuples that can be simultaneously achieved in the asymptotic block-length regime~\cite{amor2016fundamental}. To capture the trade-off between the information and energy rates,~\cite{varshney2008transporting} defines a capacity-energy function for various channels including the discrete memoryless channel, binary symmetric channel, and the AWGN channel. In~\cite{GroverSahai}, the information-energy trade-off is studied for a coupled-inductor circuit that models a  slow frequency-selective fading channel.
The information-energy capacity region of the Gaussian multiple access channel is characterized in~\cite{amor2016feedback}, whereas the information-energy capacity region of the Gaussian interference channel is approximated in~\cite{KhalfetGIC}. Nonetheless, the information-energy trade-off is not the only trade-off involved in SIET. In the finite block-length regime, several other trade-offs appear (between the information rate, energy rate, DEP and EOP) which are taken into consideration in this paper. Within the finite block-length regime,~\cite{perlaza2018simultaneous} and~\cite{khalfet2019ultra} provide a characterization of the information-energy capacity region with binary antipodal channel inputs. 
Converse and achievability bounds on SIET for a given arbitrary number of channel inputs is presented in~\cite{zuhraITW} and~\cite{zuhraISIT}, respectively. 
The impact of energy harvester non-linearities on the fundamental limits of SIET in the finite block-length regime has been studied in~\cite{zuhraITW2}.


The second area of research that has received considerable attention in SIET is the modeling of the energy harvester (EH) circuits with the aim of providing accurate estimates of the energy harevested from an RF signal. This line of inquiry has revealed that, due to the presence of non-linear elements such as diodes in the EH circuits, the expected energy harvested from a signal is a function of the fourth power of the signal magnitude~\cite{varasteh2017wireless} in addition to the squared magnitude as was conventionally assumed (see~\cite{amor2016fundamental}, \cite{amor2016feedback} and~\cite{khalfet2019ultra}). Recent research on EH non-linearities~\cite{8115220,7547357} has shown that energy models that do not account for these non-linearities result in inaccurate estimates of the harvested energy. This work accounts for the EH non-linearities by adopting the models proposed in the literature for determining the energy harvested from the transmitted RF signals. 

The third category of research deals with aspects related to the design and implementation of SIET such as signal and system design, resource allocation, receiver architectures, energy harvester circuits, and decoding strategies. Optimal waveform design for SIET from a multi-antenna transmitter to multiple single antenna receivers is studied in~\cite{9447959}. Signal and system design exclusively for wireless energy transmission has been studied in~\cite{7547357,9411899,9184149,7867826, 9447237} and~\cite{9153166}. In~\cite{9377479}, the authors optimize resource allocation and beamforming for intelligent reflecting surfaces aided SIET. The memory of non-linear elements in the EH circuit is modeled as a Markov decision process in~\cite{9241856} and a learning based model is proposed for the EH circuit.
 An algorithm for designing a circular quadrature amplitude modulation scheme for SIET that maximizes the peak-to-average power ratio has been proposed in~\cite{9593249}. This paper contributes to this line of research by providing a method of code construction for SIET that satisfies a given feasible tuple of information rate, energy rate, DEP and EOP. 
 
More comprehensive overviews of the work on SIET in the second and third categories detailed above, can be found in~\cite{survey,surveyA},~\cite{8476597}, and~\cite{ClerckxFoundations}. A comparison of relevant aspects of the existing literature on SIET with this work is provided in Table~\ref{TableSOA} below. A tick mark indicates that the specific factor has been taken into consideration in the paper while a dash indicates that the concerned feature has not been taken into consideration in that reference.

\begin{table}[!htbp] 
\captionof{table}{Summary of the state-of-art} 
\centering
\begin{tabular}{ | m{1.5 cm} | m{1.3 cm} |m{1.7 cm} |m{1.5 cm} | m{3cm}| m{1cm}| m{1cm}|} 
 \hline
  \textbf{Reference}  & \textbf{Channel inputs} & \textbf{Block-length} & \textbf{EH non-linearities} & \textbf{Channel} & \textbf{DEP} & \textbf{EOP} \\
 \hline
 \hline
\cite{varshney2008transporting}       & Infinite                    & Asymptotic			& -				& DMC, AWGN & - & -\\
\hline
\cite{GroverSahai}  		      &Infinite			 & 	Asymptotic			& 	-			&  AWGN + Frequency-selective fading & - & - \\
\hline
\cite{NizarItw2021} 			& Infinite                    &     Asymptotic                            & \checkmark 		& Rayleigh fading & - & - \\
\hline
\cite{6373669} 				& Infinite 	               & Asymptotic				& 	-			& Flat fading & - & - \\
\hline
\cite{9241856} 		                & Infinite		       	       & Asymptotic		& -				& AWGN& - & - \\
 \hline
 \cite{9149424} 		        & Finite		       	       & 	Asymptotic			& -				& AWGN + Fading& - & - \\
 \hline
 \cite{9377479} 		        & Finite		       	       & 	Asymptotic			& -				& AWGN + Flat-fading& - & - \\
 \hline
 \cite{7998252} 		        & Finite		       	              & -			& 	-			& AWGN + Rayleigh fading & \checkmark & - \\
 \hline
  \cite{9593249} 		& Finite		       	              & -			& -				& AWGN& -  & - \\
 \hline
    \cite{varasteh2017wireless} 	& Infinite		       	              & Asymptotic	& \checkmark			& AWGN& - & - \\
 \hline
     \cite{8115220} 		        & Infinite		       	              & Asymptotic	& \checkmark			& Multi-path fading & - & - \\
     \hline
   \cite{9741251}			&	Uncountable	&	Asymptotic	&	-		&	AWGN + Rayleigh fading	& \checkmark & -\\
  \hline
   \cite{9734045}			&	Uncountable	&	Asymptotic	& -		&	AWGN + Rayleigh and Rician fading	& - & - \\
  \hline
   \cite{liu2022joint}			&	Uncountable	&	Asymptotic	&	-		&	AWGN	& -& -\\
   \hline
   \cite{9502719}			&	Uncountable	&	Asymptotic	&	\checkmark		&	Multi-path fading	& - & -\\     
 \hline
    \cite{6489506}			&	Uncountable	&	Asymptotic	&	-		&	Multi-path fading	& - & -\\     
 \hline
     \cite{ 7063588}			&	Uncountable	&	Asymptotic	&	-		&	AWGN + Flat-fading	& - & -\\     
 \hline
      \cite{9169700 }			&	Finite	&	Finite	&	 \checkmark			&	AWGN 	& - & -\\     
 \hline
  \cite{perlaza2018simultaneous}               & Finite                         & Finite              & 	-			& BSC & \checkmark & \checkmark\\
\hline
 \cite{khalfet2019ultra} 	& Finite		       	              & Finite			& 	-			& BSC& - & -\\
 \hline
      \textbf{This work} 	& Finite		       	              & Finite	& \checkmark			& AWGN& \checkmark & \checkmark\\
        \hline
 \end{tabular} \label{TableSOA}
 \end{table}

\subsection{Contributions}
This work studies the fundamental limits of the joint transmission of information and energy in the finite block-length regime for a given finite and discrete set of channel input symbols (\ie, constellations, such as QPSK, 16-QAM, 64-QAM etc).
In the current literature, no techniques exist for analyzing such a setup. 
The analysis in this work breaks away from the traditional method of taking a DEP and then bounding the information rate in terms of the DEP. Instead, in a new technique developed here, all the parameters of concern, namely, the energy transmission rate, the information transmission rate, DEP and EOP are evaluated as a function of the type. A type is understood in the sense of  the empirical frequency with which each channel input symbol appears in the codewords~\cite{CsiszarMoT}. The trade-offs between the energy transmission rate, the information transmission rate, DEP and EOP can then be easily analyzed through the type of the code. 
%

The main contributions of this work are summarized below.
\begin{itemize}
\item A new technique for analyzing the fundamental limits of SIET in the finite block-length regime with finite and discrete set of channel inputs is developed using the method of types~\cite{CsiszarMoT}. It is shown that the bounds on all the parameters of SIET can be expressed as a function of the type induced by the code. Therefore, the trade-offs between these parameters can be controlled by adjusting the type. 

It is noteworthy that though this technique is developed in the context of SIET, its applicability is not limited to SIET. The technique   provides an elegant method of defining the fundamental limits in any information theoretic problem where the parameters under consideration can be expressed as a function of the type induced by the code.
\item Impossibility bounds for finite block-length SIET over an AWGN channel with a peak-amplitude constraint are characterized for a given finite set of channel input symbols. Impossibility bounds define the set of information rate, energy rate, DEP and EOP tuples that cannot be simultaneously achieved by any code that employs the given set of channel input symbols.
\item Based on the insights provided by the impossibility results, a method of constructing codes for SIET over an AWGN channel in the finite block-length regime is proposed.
\item The achievable information-energy region for the constructed family of codes is characterized. The achievable region defines the set of information rate, energy rate, DEP and EOP tuples that can be simultaneously achieved by at least one code from the constructed family of codes.
\item The trade-offs between the information rate, energy rate, DEP and EOP are illustrated through various instructive examples. The information-energy region is also illustrated for a given set of channel input symbols. The study of the gaps between the impossibility and achievability bounds shows that the bounds match in terms of the information rate, the energy rate and the EOP. For a given information rate, energy rate and EOP, however, the achieved DEP  is higher than the impossibility bound due to the choice of the decoding regions made during the construction.
\end{itemize}

The notation used in this paper is summarized below.
  
\subsection{Notation} \label{SecNotation}
The sets of natural, real and complex numbers are denoted by $\ints$, $\reals$ and $\complex$, respectively. In particular, $0 \notin \ints$.
Random variables and random vectors are denoted by uppercase letters and uppercase bold letters, respectively. Scalars are denoted by lowercase letters and vectors by lowercase bold letters. The real and imaginary parts of a complex number $c \in \complex$ are denoted by  $\Re(c)$ and $\Im(c)$, respectively. The complex conjugate of $c \in \complex$ is denoted by $c^\star$ and the magnitude of $c$ is denoted by $|c|$.  The imaginary unit is denoted by $\mathrm{i}$, \ie, $\mathrm{i}^2 = -1$.  
The empty set is denoted by $\phi$. The $\sinc$ function is defined as follows
\begin{IEEEeqnarray}{rCl} \label{EqSinc}
\sinc(t) \triangleq \frac{\sin(\pi t)}{\pi t},
\end{IEEEeqnarray}
and the $\mathrm{Q}$ function is given by the following:
\begin{IEEEeqnarray}{rCl} \label{DefQfunc}
\mathrm{Q}(x) = \int_x^\infty \frac{1}{\sqrt{2 \pi}} \exp \left(-\frac{t^2}{2} \right) \mathrm{d}t.
\end{IEEEeqnarray}
The rest of this paper is organized as follows. The system model is presented in Section~\ref{sec:system_model}. The impossibility results for finite block-length SIET are characterized in Section~\ref{sec:results}. Section~\ref{SecAchievability} provides the construction of codes for finite block-length SIET followed by the characterization of an achievable region. The trade-offs between various parameters of finite block-length SIET are elucidated using various examples in Section~\ref{SecDiscussion}. Section~\ref{SecDiscussion} also provides a comparison of the proposed bounds with the existing state of the art.

\section{System Model}  \label{sec:system_model}
Consider a communication system formed by a transmitter, an information receiver (IR), and an energy harvester (EH). The objective of the transmitter is to simultaneously send information to the IR at a rate of $R$ bits per second; and energy to the EH at a rate of $B$ Joules per second over an AWGN channel. The transmission takes place over a finite duration of $n \in \ints$ channel uses. The transmitter uses $L$ symbols from the set of channel input symbols
\begin{equation}\label{EqCIsymbols}
\mathcal{X} \triangleq \{x^{(1)}, x^{(2)}, \ldots, x^{(L)}\} \subset \mathds{C},
\end{equation} 
where,
\begin{equation}\label{EqL}
L \triangleq \left| \mathcal{X} \right|.
\end{equation}
For all $m \inCountK{n}$, denote by $\nu_m \in \mathcal{X}$, the symbol to be transmitted during channel use $m$.
Denote the vector of channel input symbols transmitted over $n$ channel uses by
\begin{IEEEeqnarray}{rCl} \label{Eqnu}
\boldsymbol{\nu} = (\nu_1, \nu_2, \ldots, \nu_n )^{\sf{T}}.
\end{IEEEeqnarray}
The baseband frequency of the transmitter in Hertz (Hz) is denoted by $f_w$. Denote by $T = \frac{1}{f_w}$, the duration of a channel use in time units. Hence, the transmission takes place over $nT$ time units.
The complex baseband signal at time $t$, with $t \in [0,nT]$ is given by
\begin{IEEEeqnarray}{rCl}
x(t) &=& \sum_{m=1}^n  \nu_m \sinc \left( f_w \left(t - (m-1)T\right) \right), \label{Eq4} 
\end{IEEEeqnarray}
where the $\sinc$ function is in~\eqref{EqSinc}. The signal $x(t) $ in~\eqref{Eq4} has a bandwidth of $\frac{f_w}{2} > 0$ Hz. Let $f_c > \frac{f_w}{2}$ denote the center frequency of the transmitter. The RF signal input to the channel at time $t$, denoted by $\tilde{x}(t)$, is obtained by the frequency up-conversion of the baseband signal $x(t)$ in~\eqref{Eq4} as follows:
\begin{IEEEeqnarray}{rcl}
\tilde{x}(t) &=& \Re \left(x(t) \sqrt{2} \exp(\mathrm{i} 2 \pi f_c t) \right), \label{EqXrf} 
\end{IEEEeqnarray}
where $\mathrm{i}$ is the complex unit.
The RF outputs of the AWGN channel at time $t \in [0,nT]$ are the random variables
\begin{subequations}
\begin{IEEEeqnarray}{rCl}
   \label{EqYt} Y(t) & = & \tilde{x}(t) + N_1(t), \mbox{ and }  \\
    Z(t) & = & \tilde{x}(t) + N_2(t), \label{Eq7b}
\end{IEEEeqnarray}
\end{subequations}
where, for all $t \in [0,nT]$, the random variables $N_1(t)$ and $N_2(t)$ represent real white Gaussian noise with zero mean and variance $\sigma^2$; $Y(t)$ and $Z(t)$ are the inputs to the IR and the EH, respectively. 

At the IR, the received signal $Y(t)$ in~\eqref{EqYt} is first multiplied with $\sqrt{2} \exp(-\mathrm{i} 2 \pi f_c t)$ to obtain the down-converted output. The down-converted output is then passed through a unit gain low pass filter with impulse response $f_w \sinc \left( f_w t \right)$ that has a cut-off frequency of $\frac{f_w}{2}$ Hz to obtain the complex baseband equivalent of $Y(t)$. This is followed by ideally sampling the complex baseband output at intervals of $1/f_w$. 
The resulting discrete time baseband output at the end of $n$ channel uses is given by the following random vector~\cite[Section $2.2.4$]{tseV}:
\begin{IEEEeqnarray}{rCl} \label{EqChannelModel}
{\boldsymbol Y} & = & \boldsymbol{\nu} + {\boldsymbol N},
\end{IEEEeqnarray}
where the vector ${\boldsymbol Y} = (Y_1,Y_2, \ldots, Y_n)^{\sf{T}} \in \mathds{C}^{n}$ is the input to the IR; $\boldsymbol{\nu}$ is the vector of channel input symbols in~\eqref{Eqnu}; and $\boldsymbol{N} = (N_{1}, N_{2}, \ldots, N_{n})^{\sf{T}}\in \mathds{C}^{n}$ is the noise vector such that, for all $m \inCountK{n}$, the random variable $N_{m}$ is a complex circularly symmetric Gaussian random variable whose real and imaginary parts have zero means and variances $\frac{1}{2}\sigma^2$. Moreover, the random variables $N_1, N_2, \ldots, N_n$ are mutually independent (see~\cite[Section $2.2.4$]{tseV}). That is, for all $\boldsymbol{y} = (y_1, y_2, \ldots, y_n)^{\sf{T}} \in \mathds{C}^{n}$, and all $\boldsymbol{\nu} = (\nu_1, \nu_2, \ldots, \nu_n)^{\sf{T}} \in \mathds{C}^{n}$, the conditional probability density function of the channel output $\boldsymbol{Y}$ in~\eqref{EqChannelModel} is given by
\begin{IEEEeqnarray}{rCl} \label{EqYXdistribution}
    f_{\boldsymbol{Y}|\boldsymbol{X}}(\boldsymbol{y}|\boldsymbol{x}) & = & \prod_{m=1}^n f_{Y|X}(y_m|\nu_m),
\end{IEEEeqnarray}
where, for all $m \in \lbrace 1,2, \ldots, n \rbrace$,
\begin{IEEEeqnarray}{rCl}
\label{Eq4a}
f_{Y|X}(y_m|\nu_m) &=& \    \frac{1}{\pi \sigma^2}\exp \left( - \frac{\left| y_m - \nu_m \right|^2}{\sigma^2} \right) \\
 \label{EqDensities}
 &=&   \frac{1}{\pi \sigma^2} \exp \left(- \frac{(\Re(y_m)-\Re(\nu_m))^2 +(\Im(y_m) - \Im(\nu_m))^2}{\sigma^2} \right).
 \end{IEEEeqnarray}
The EH does not down-convert or filter the received input $Z(t)$ (see~\cite{7547357} and~\cite{8115220}). The RF signal $Z(t)$ in~\eqref{Eq7b} is used as is for harvesting the energy contained in it.

Within this framework, the two tasks of information and energy transmission must be simultaneously accomplished.

\subsection{Information and Energy Transmission} \label{SubsecInformationEnergyTransmission}
%
Let $M$ be the cardinality of the set of indices from which a message is transmitted within $n$ channel uses. That is, 
\begin{equation} \label{eq:M_L}
  M \leqslant 2^{n \log L},
\end{equation}
where $L$ is the number of channel input symbols in~\eqref{EqL}. To accomplish the tasks of information and energy transmission, the transmitter makes use of a code for representing these $M$ messages. The encoding, decoding and energy harvesting operations are defined as follows.
\subsubsection{Encoding} 
To reliably transmit a message index, the transmitter uses an $(n,M,\mathcal{X},P)$-code defined as follows.
\begin{definition}[$(n,M,\mathcal{X},P)$-code] \label{DefNmCode}
An $(n,M,\mathcal{X},P)$-code for the random transformation in~\eqref{EqChannelModel} with the set of channel input symbols $\mathcal{X} \subset \complex$ in~\eqref{EqCIsymbols} is a system:
\begin{equation}
    \left \lbrace ({\boldsymbol u}(1),\mathcal{D}_1), ({\boldsymbol u}(2),\mathcal{D}_2), \ldots, ({\boldsymbol u}(M),\mathcal{D}_M)\right \rbrace,
\end{equation}
where, for all $(i,j) \in \{1,2, \ldots, M\}^2, i\neq j$ and all $m \inCountK{n}$,
\begin{subequations}\label{EqCodeProperties}
\begin{align}
\label{eq:u_i}  &{\boldsymbol u}(i) = (u_1(i), u_2(i), \ldots, u_n(i)) \in \mathcal{X}^n, \\
        &\mathcal{D}_i \cap \mathcal{D}_j = \phi,\\
        &\bigcup_{i = 1}^M \mathcal{D}_i \subseteq \mathds{C}^n, \\
        & \left| \mathcal{X} \right| = L, \mbox{ and }\\
  \label{EqPCriteria}      &|u_m(i)| \leqslant P,
    \end{align}
\end{subequations}
where the real $P > 0$ is the peak-amplitude constraint.
\end{definition}
Assume that the transmitter uses the $(n,M,\mathcal{X},P)$-code
\begin{equation} \label{Eqnm_code}
    \mathscr{C} \triangleq \{({\boldsymbol u}(1),\mathcal{D}_1), ({\boldsymbol u}(2),\mathcal{D}_2), \ldots, ({\boldsymbol u}(M),\mathcal{D}_M)\},
\end{equation}
that satisfies~\eqref{EqCodeProperties}.
The results in this paper are presented in terms of the types induced by the codewords of the given code $\mathscr{C}$ in~\eqref{Eqnm_code}. The type induced by the codeword $\boldsymbol{u}(i)$, with $i \inCountK{M}$, is a probability mass function (pmf) whose support is equal to or a subset of $\mathcal{X}$ in~\eqref{EqCIsymbols}. This pmf is denoted by $P_{\boldsymbol{u}(i)}$ and for all $x \in \mathcal{X}$,
\begin{equation} \label{eq:u_measure}
    P_{\boldsymbol{u}(i)}(x) \triangleq \frac{1}{n} \sum_{m=1}^n \mathds{1}_{\{u_m(i) = x\}}.
\end{equation}
The type $P_{\boldsymbol{u}(i)}$ captures the empirical codeword distribution~\cite{polyanskiy2013empirical} of code $\mathscr{C}$.

The type induced by all the codewords in $\mathscr{C}$ is also a pmf on the set $\mathcal{X}$ in~\eqref{EqCIsymbols}. This pmf is denoted by  $P_{\mathscr{C}}$ and  for all $x \in \mathcal{X}$,
\begin{equation} \label{eq:p_bar}
    P_{\mathscr{C}}(x) \triangleq \frac{1}{M} \sum_{i=1}^M P_{\boldsymbol{u}(i)}(x).
\end{equation}
Using~\eqref{eq:u_measure} and~\eqref{eq:p_bar}, a class of codes called constant composition codes can be defined as follows.  
\begin{definition}[Constant Composition Codes]\label{DefHC}
An $(n,M,\mathcal{X},P)$-code $\mathscr{C}$ for the random transformation in~\eqref{EqChannelModel} of the form in~\eqref{Eqnm_code} is said to have constant composition if for all $i \inCountK{M}$ and for all $x \in \mathcal{X}$, with $\mathcal{X}$ in~\eqref{EqCIsymbols} it holds that
\begin{equation}\label{EqHomogeneousCodes}
 P_{\boldsymbol{u}(i)}(x) =  P_{\mathscr{C}}(x),
\end{equation}
where, $P_{\boldsymbol{u}(i)}$ and $P_{\mathscr{C}}$ are the types defined in~\eqref{eq:u_measure} and~\eqref{eq:p_bar}, respectively.
\end{definition}
Therefore, constant composition codes are $(n,M,\mathcal{X},P)$-codes in which  a given channel input symbol is used the same number of times in all codewords. Constant composition codes are of interest to this study because, as we will see in the following sections, several upper and lower bounds are either simplified or made tractable for the  case of constant composition codes.

The information transmission rate of any $(n,M,\mathcal{X},P)$-code $\mathscr{C}$ is given by
\begin{equation} \label{EqR}
    R(\mathscr{C}) = \frac{\log_2 M}{n}
\end{equation}
in bits per channel use.
\subsubsection{Information Decoding}
To transmit the message index $i$, with $i \in \lbrace 1,2, \ldots, M \rbrace$, the transmitter uses the codeword ${\boldsymbol u}(i)= (u_1(i), u_2(i), \ldots, u_n(i))$. That is, at channel use $m$, with $m \in \lbrace 1,2, \ldots, n\rbrace$, the transmitter inputs the RF signal corresponding to symbol $u_{m}(i)$ into the channel. At the end of $n$ channel uses, the IR observes a realization of the random vector ${\boldsymbol Y} = (Y_1, Y_2, \ldots, Y_n)^{\sf{T}}$ in~\eqref{EqChannelModel}.
Let $W$ be a random variable that denotes the message index transmitted. For all $i \inCountK{M}$, the probability of transmitting message index $i$ is given by
\begin{IEEEeqnarray}{rCl} \label{Eqprior}
P_W \left(i \right) = \frac{1}{M}.
\end{IEEEeqnarray}
 The IR decides that message index $i$, with $i \inCountK{M}$, was transmitted, if the following event takes place:
\begin{equation}
    {\boldsymbol Y} \in \mathcal{D}_i,
\end{equation}
with $\mathcal{D}_i$ in~\eqref{Eqnm_code}.
For each $i \inCountK{M}$, the set $\mathcal{D}_i \subseteq \mathds{C}^n$ that minimizes the DEP for codeword $\boldsymbol{u}(i)$ is defined by the maximum a posteriori (MAP) decision rule~\cite[Chapter $21$]{lapidoth} as follows
\begin{IEEEeqnarray}{rCl} \label{EqMAPregion}
\mathcal{D}_i &=& \lbrace \boldsymbol{y} \in \complex : \forall \ i' \inCountK{M}, P_W \left(i\right) f_{\boldsymbol{Y}|\boldsymbol{X}}(\boldsymbol{y}|\boldsymbol{u}(i)) \geq P_W \left(i' \right) f_{\boldsymbol{Y}|\boldsymbol{X}}(\boldsymbol{y}|\boldsymbol{u}(i')) \rbrace \\
&=& \lbrace \boldsymbol{y} \in \complex : \forall \ i' \inCountK{M}, f_{\boldsymbol{Y}|\boldsymbol{X}}(\boldsymbol{y}|\boldsymbol{u}(i)) \geq f_{\boldsymbol{Y}|\boldsymbol{X}}(\boldsymbol{y}|\boldsymbol{u}(i')) \rbrace, \label{EqMAPregion2}
\end{IEEEeqnarray}
where~\eqref{EqMAPregion2} follows due to~\eqref{Eqprior}.
Therefore, the DEP associated with the transmission of message index $i$ is given by
\begin{IEEEeqnarray}{rCl}
\label{EqDEPi}
    \gamma_i(\mathscr{C}) 
    &\triangleq& 1 - \int_{\mathcal{D}_i} f_{\boldsymbol{Y}|\boldsymbol{X}}(\boldsymbol{y}|\boldsymbol{u}(i)) \mathrm{d}\boldsymbol{y},
\end{IEEEeqnarray}
and the average DEP for code $\mathscr{C}$ is given by
\begin{IEEEeqnarray}{rCl} 
    \label{eq:gamma} 
    \gamma(\mathscr{C}) &\triangleq& \frac{1}{M}\sum_{i = 1}^M \gamma_i(\mathscr{C}). 
\end{IEEEeqnarray}

%
%

%
\subsubsection{Energy Harvesting}
The channel output observed at the EH while transmitting message $i \inCountK{M}$ is $Z_{i}(t)$ in~\eqref{Eq7b}, with $t \in [0,nT]$. From~\eqref{EqXrf} and~\eqref{Eq7b}, the channel output $Z_{i}(t)$  is given by
\begin{IEEEeqnarray}{rCl} 
Z_{i}(t) &=&  \Re \Big(\sqrt{2} \sum_{m=1}^n u_m(i) \sinc \left( f_w (t - (m-1)T) \right) \exp \left(\mathrm{i} 2 \pi f_c t \right) \Big) + N_2(t) \label{Eq37} \\
&=& x_{i}(t) + N_2(t), \label{Eq38}
\end{IEEEeqnarray}
where, for all $m \inCountK{n}$, the complex $u_m(i)$ is the $m$\ts{th} symbol of the codeword $\boldsymbol{u}(i)$ in~\eqref{eq:u_i}; for all $t \in [0,nT]$, the signal $x_{i}(t)$ in~\eqref{Eq38} is 
\begin{IEEEeqnarray}{rCl}
 x_{i}(t) &=& \Re \Big(\sqrt{2} \sum_{m=1}^n u_m(i) \sinc \left( f_w (t - (m-1)T) \right) \exp \left(\mathrm{i} 2 \pi f_c t \right) \Big); 
 \end{IEEEeqnarray}
 and the random variable $N_2(t)$ is a real Gaussian random variable with zero mean and variance $\sigma^2$ in~\eqref{EqYXdistribution}.
 For all $t \in [0,nT]$, the channel output $Z_i(t)$ in~\eqref{Eq38} is a real Gaussian random variable with mean $x_i(t)$ and variance $\sigma^2$.
 
The non-linear energy model in~\cite{8115220} and~\cite{7547357} states that the energy harvested from a signal is proportional to the DC component of the second and fourth powers of the signal. Using this model, for all $i\inCountK{M}$, the expected energy harvested from the channel output $Z_i(t)$ in~\eqref{Eq37} during the time $t \in [0,nT]$ is given by the following: 
\begin{IEEEeqnarray}{rCl} \label{Eqei}
e_i &\triangleq& k_1 \sum_{m=1}^n \left| u_m(i) \right|^2 + k_2 \sum_{m=1}^n \left| u_m(i) \right|^4 \\
&=& k_1 \sum_{x \in \mathcal{X}} n P_{\boldsymbol{u}(i)} \left( x \right) \left| x \right|^2 + k_2 \sum_{x \in \mathcal{X}} n P_{\boldsymbol{u}(i)} \left( x \right) \left| x \right|^4, \label{Eq27c}
\end{IEEEeqnarray}
where $\mathcal{X}$ is in~\eqref{EqCIsymbols}; $k_1$ and $k_2$ are positive constants, with $k_1 = 0.0034$ and $k_2 = 0.3829$~\cite{8115220}; for all $m \inCountK{n}$, the complex $u_m(i)$ is in~\eqref{eq:u_i} and $P_{\boldsymbol{u}(i)}$ is the type defined in~\eqref{eq:u_measure}. 

 From~\eqref{Eq27c}, for a constant composition code $\mathscr{C}$ (Definition~\ref{DefHC}), the expected energy harvested from the channel output $Z_i(t)$ in~\eqref{Eq38} during the time $t \in [0,nT]$ is equal for all $i \inCountK{M}$. That is, for all $i \inCountK{M}$,
 \begin{subequations}\label{EqeiHomogeneous}
 \begin{IEEEeqnarray}{rCl} 
 e_i = e_\mathscr{C},
  \end{IEEEeqnarray}
  where,
 \begin{IEEEeqnarray}{rCl} \label{EqeiHomogeneousb}
e_\mathscr{C} \triangleq k_1 \sum_{x \in \mathcal{X}} n P_{\mathscr{C}} \left( x \right) \left| x \right|^2 + k_2 \sum_{x \in \mathcal{X}}n P_{\mathscr{C}} \left( x \right) \left| x \right|^4,
 \end{IEEEeqnarray}
 \end{subequations}
with $\mathcal{X}$ in~\eqref{EqCIsymbols} and $P_{\mathscr{C}}$ is the type defined in~\eqref{eq:p_bar}.
 
Denote by $M' \leq M$, the number of unique values in the vector $\left(e_1, e_2, \ldots, e_M \right)^{\sf{T}}$ with $e_i$ in~\eqref{Eqei}. The $M'$ unique energy levels are given by $\left\lbrace \bar{e}_1, \bar{e}_2, \ldots, \bar{e}_{M'} \right\rbrace$. Assume without loss of generality that the following holds:
\begin{IEEEeqnarray}{rCl} \label{EqUniqueLevels}
0 < \bar{e}_1 < \bar{e}_2 < \ldots < \bar{e}_{M'}.
\end{IEEEeqnarray}
More specifically, for all $i \inCountK{M}$, there exists $j \inCountK{M'}$ such that $e_i = \bar{e}_j$, with $e_i$ in~\eqref{Eqei}. 

For all $j \inCountK{M'}$, define variables $y_j$ to be the number of codewords that carry energy $\bar{e}_j$. More precisely, for all $j \inCountK{M'}$, $y_j$ is given by
\begin{IEEEeqnarray}{rCl} \label{EqYindicator}
y_j = \sum_{i=1}^M \mathds{1}_{\left\lbrace e_i = \bar{e}_j \right\rbrace }.
\end{IEEEeqnarray}
It follows that $\sum_{j=1}^{M'} y_j = M$. 

The energy harvested during time $t \in [0, nT]$ is a random variable denoted by $E$.
The probability of harvesting energy $e$ given that message index $i$ was transmitted is given by
\begin{IEEEeqnarray}{rCl} 
P_{E | W} \left(e|i \right) = \mathds{1}_{\lbrace e = e_i \rbrace},
\end{IEEEeqnarray}
where $W$ is the random variable that denotes the transmitted message index in~\eqref{Eqprior} and $e_i$ is in~\eqref{Eqei}. The pmf of the random variable $E$, denoted by $P_E$ is given by the following:
\begin{IEEEeqnarray}{rCl} \label{Eq35}
P_{E} \left(e\right) &=& \sum_{i=1}^M P_{E | W} \left(e|i \right) P_W \left(i \right) \\
&=& \frac{1}{M} \sum_{i=1}^M \mathds{1}_{\lbrace e = e_i \rbrace}. \label{Eq35b}
\end{IEEEeqnarray}
Let $B \in \reals$ be the energy rate required at the EH. Then, using~\eqref{Eq35}, the EOP associated with code $\mathscr{C}$ is defined as follows:
\begin{IEEEeqnarray}{rCl} \label{eq:theta_def}
 \theta(\mathscr{C},B) \triangleq \mathrm{Pr} \left( E<B  \right) &=&  \sum_{i \in \lbrace j \inCountK{M}: e_j < B \rbrace}  P_E (e_i) \\
 &=& \sum_{i \in \lbrace j \inCountK{M}: e_j < B \rbrace} \frac{1}{M} \sum_{i=1}^M \mathds{1}_{\lbrace e = e_i \rbrace} \label{Eq37b} \\
  &=& \frac{1}{M} \Big| \lbrace i \inCountK{M}: e_i<B \rbrace \Big| \\
  & = & \frac{1}{M} \sum_{i=1}^M \mathds{1}_{\lbrace e_i < B \rbrace}, \label{Eq25}
 \end{IEEEeqnarray}
 where, $P_E$ is the pmf defined in~\eqref{Eq35}, $e_i$ is defined in~\eqref{Eqei} and, the equality in~\eqref{Eq37b} follows from~\eqref{Eq35b}. 

From~\eqref{Eq25}, an interesting insight can be derived about $\theta(\mathscr{C},B)$. The EOP $\theta(\mathscr{C},B)$ can only take discrete values in the set $\lbrace 0, \frac{1}{M}, \frac{2}{M}, \ldots, 1 \rbrace$. Moreover, for constant composition codes, it holds that $\theta(\mathscr{C},B) \in \lbrace 0,1 \rbrace$. 

Definition~\ref{DefNmCode} can now be refined to include the DEP and the EOP as follows.
\begin{definition}[$(n,M,\mathcal{X},P,\epsilon,B,\delta)$-code] \label{def:nmed_code}
An $(n,M,\mathcal{X},P)$-code $\mathscr{C}$ for the random transformation in~\eqref{EqChannelModel} is said to be an $(n,M,\mathcal{X},P,\epsilon,B,\delta)$-code  if  the following hold:
\begin{IEEEeqnarray}{rCl}
\label{EqGammaUpperbound}
\gamma(\mathscr{C}) &\leq& \epsilon, \mbox{and} \\
\label{eq:delta}
    \theta(\mathscr{C},B) &\leq& \delta.
\end{IEEEeqnarray}
\end{definition}
%
%
\section{Impossibility Bounds} \label{sec:results}
This section characterizes the \emph{impossibility region} of $(n,M,\mathcal{X},P,\epsilon,B,\delta)$-codes. The term \emph{impossibility region} is used to refer to the set of all information rate, energy rate, DEP, EOP tuples $(R,B,\epsilon,\delta) \in \reals^2 \times [0,1]^2$ that are not achievable by any code that uses a given set of channel input symbols $\mathcal{X}$.
In the interest of clarity, we refrain from using the term \emph{`converse'} to refer to these bounds since a converse is understood to be independent of the channel inputs.
However, determining the optimal set of channel inputs is known to be a very difficult problem~\cite{8613368,8878162,8849318}, even without energy considerations. 
The process of characterizing the information-energy regions for given sets of channel input symbols reveals the dependence of the information rate, the energy rate, the DEP and the EOP on various parameters of the code which, in turn, provides insights into which parameters should be changed and how in order to achieve a certain performance objective in SIET.

In the following subsection, a lower bound on the DEP and an upper bound on the information transmission rate for an $(n,M,\mathcal{X},P,\epsilon,B,\delta)$-code are characterized.
 \subsection{Information Transmission Rate and DEP} \label{SubsecConverseR}
Given the channel input $\boldsymbol{u}(i) = (u_1(i), u_2(i), \ldots, u_n(i)) \in \mathds{C}^{n}$ in~\eqref{eq:u_i}, the channel output vector $\boldsymbol{Y} = (Y_1,Y_2, \ldots, Y_n)^{\sf{T}} \in \mathds{C}^{n}$ in~\eqref{EqChannelModel} is a complex Gaussian vector with mean $\boldsymbol{u}(i)$. Therefore, the conditional probability density function of $\boldsymbol{Y}$ is given by
\begin{IEEEeqnarray}{rCl} 
f_{\boldsymbol{Y}|\boldsymbol{X}}(\boldsymbol{y}|\boldsymbol{u}(i)) & = & \frac{1}{\left(\pi \sigma^2\right)^n}\exp \left( - \frac{\left|\boldsymbol{y} - \boldsymbol{u}(i) \right|^2}{\sigma^2} \right).
\end{IEEEeqnarray}
It is known that the DEP for this vector Gaussian channel is minimized by the MAP decoder~\cite[Chapter $21$]{lapidoth} which decides that message $i$ was transmitted if the channel output vector $\boldsymbol{Y} \in \mathcal{D}_i$ in~\eqref{EqMAPregion}. That is,
\begin{IEEEeqnarray}{rCl} \label{Eq57}
i \in \arg \max_{i' \inCountK{M}} f_{\boldsymbol{Y}|\boldsymbol{X}}(\boldsymbol{y}|\boldsymbol{u}(i')).
\end{IEEEeqnarray}
By simple mathematical manipulations, the decision rule in~\eqref{Eq57} simplifies to the minimum distance decoder~\cite[Chapter 4]{proakis} which decides that message $i$ was transmitted if
\begin{IEEEeqnarray}{rCl} \label{Eq58b}
i \in \arg \min_{i' \inCountK{M}} \left| \boldsymbol{y} - \boldsymbol{u}(i') \right|.
\end{IEEEeqnarray}
Using~\eqref{Eq58b}, the following theorem provides a lower bound on the DEP for a constant composition $(n,M,\mathcal{X},P,\epsilon,B,\delta)$-code $\mathscr{C}$.
%
%
\begin{theorem} \label{LemmaImpossibleDEP}
Given a constant composition $(n,M,\mathcal{X},P,\epsilon,B,\delta)$-code $\mathscr{C}$ for the random transformation in~\eqref{EqChannelModel} of the form in~\eqref{Eqnm_code}, for all $\ell \inCountK{L}$, the complex $\bar{x}^{(\ell)} \in \mathcal{X}$ in~\eqref{EqCIsymbols} is given by
\begin{IEEEeqnarray}{rCl} \label{EqNeighbor}
\bar{x}^{(\ell)} \in \arg\max_{x \in \mathcal{X} \setminus \lbrace x^{(\ell)} \rbrace} \left|x^{(\ell)} - x \right|.
\end{IEEEeqnarray}
Then, the DEP $\epsilon$ satisfies the following:
\begin{IEEEeqnarray}{rCl} \label{EqboundDEP}
\epsilon \geq (M-1) \mathrm{Q} \left( \sqrt{\frac{\sum_{\ell=1}^L n P_{\mathscr{C}}(x^{(\ell)}) \left| x^{(\ell)} - \bar{x}^{(\ell)} \right|^2}{2 \sigma^2}} \right),
\end{IEEEeqnarray}
where, $P_{\mathscr{C}}$ is the type defined in~\eqref{eq:p_bar}; $L$ is the number of symbols in~\eqref{EqL}, and the real $\sigma^2$ is the noise variance in~\eqref{Eq4a}.
\end{theorem}
\begin{proof}
The decoding error probability $\gamma_i \left( \mathscr{C} \right)$ in~\eqref{EqDEPi} for the code $\mathscr{C}$, given that the message index $i \inCountK{M}$ is transmitted is
{\allowdisplaybreaks
\begin{IEEEeqnarray}{rCl} 
\gamma_i \left( \mathscr{C} \right) &=& \mathrm{Pr} \left( \boldsymbol{Y} \notin \mathcal{D}_i \big| \boldsymbol{u}(i) \right) \label{Eq662} \\
&=& \sum_{\substack{j =1 \\ j\neq i}}^M \mathrm{Pr} \left( \left| \boldsymbol{Y} - \boldsymbol{u}(j) \right|^2 < \left| \boldsymbol{Y} - \boldsymbol{u}(i) \right|^2 \big| \boldsymbol{u}(i) \right)  \label{Eq663} \\
&=& \sum_{\substack{j =1 \\ j\neq i}}^M \mathrm{Pr} \left( \left| \boldsymbol{u}(i) + \boldsymbol{N} - \boldsymbol{u}(j) \right|^2 < \left| \boldsymbol{u}(i) + \boldsymbol{N} - \boldsymbol{u}(i) \right|^2 \right) \\
&=& \sum_{\substack{j =1 \\ j\neq i}}^M \mathrm{Pr} \left( \left| \boldsymbol{u}(i) + \boldsymbol{N} - \boldsymbol{u}(j) \right|^2 < \left|\boldsymbol{N} \right|^2 \right) \\
&=& \sum_{\substack{j =1 \\ j\neq i}}^M \mathrm{Pr} \left( \sum_{m=1}^n \left| u_m(i) + N_m - u_m(j) \right|^2 < \sum_{m=1}^n \left| N_m \right|^2 \right) \\
&=& \sum_{\substack{j =1 \\ j\neq i}}^M \mathrm{Pr} \left( \sum_{m=1}^n (\Re(u_m(i) + N_m - u_m(j)))^2 + (\Im(u_m(i) + N_m - u_m(j)))^2 < \sum_{m=1}^n \Re(N_m)^2 + \Im(N_m)^2 \right) \\
&=& \sum_{\substack{j =1 \\ j\neq i}}^M \mathrm{Pr} \Big( \sum_{m=1}^n (\Re(u_m(i)) - \Re(u_m(j)))^2 + (\Im(u_m(i)) - \Im(u_m(j)))^2 \nonumber \\
&& + 2 \Re(N_m)(\Re(u_m(i)) - \Re(u_m(j))) + 2 \Im(N_m)(\Im(u_m(i)) - \Im(u_m(j))) < 0 \Big) \\
&=& \sum_{\substack{j =1 \\ j\neq i}}^M \mathrm{Pr} \Big( \sum_{m=1}^n \Re(N_m)(\Re(u_m(i)) - \Re(u_m(j))) + \Im(N_m)(\Im(u_m(i)) - \Im(u_m(j))) \nonumber \\
&& < - \frac{1}{2} \sum_{m=1}^n (\Re(u_m(i)) - \Re(u_m(j)))^2 + (\Im(u_m(i)) - \Im(u_m(j)))^2 \Big), \label{Eq68b}
\end{IEEEeqnarray}
}
where, the equality in~\eqref{Eq663} follows from~\eqref{Eq58b}, and $\boldsymbol{N} = (N_{1}, N_{2}, \ldots, N_{n})^{\sf{T}}$ is the AWGN noise vector in~\eqref{EqChannelModel} such that, for all $m \inCountK{n}$, the random variable $N_{m}$ is a complex circularly symmetric Gaussian random variable whose real and imaginary parts have zero means and variances $\frac{1}{2}\sigma^2$. 
Therefore, the random variable 
\begin{IEEEeqnarray}{rCl} \label{Eq69}
\overline{N}_{i,j} = \sum_{m=1}^n \Re(N_m)(\Re(u_m(i)) - \Re(u_m(j))) + \Im(N_m)(\Im(u_m(i)) - \Im(u_m(j)))
\end{IEEEeqnarray}
in~\eqref{Eq68b} is a linear combination of $2n$ Gaussian random variables, each with mean zero and variance $\frac{1}{2}\sigma^2$. Thus, $\overline{N}_{i,j}$ is also a zero mean Gaussian random variable with variance given by  
\begin{IEEEeqnarray}{rCl}  \label{Eq700}
\sigma_{\overline{N}_{i,j}}^2 =  \frac{\sigma^2}{2} \sum_{m=1}^n (\Re(u_m(i)) - \Re(u_m(j)))^2 + (\Im(u_m(i)) - \Im(u_m(j)))^2.
\end{IEEEeqnarray}
From~\eqref{Eq68b},\eqref{Eq69} and~\eqref{Eq700}, it follows that,
\begin{IEEEeqnarray}{rCl} \label{Eq72}
\gamma_i \left( \mathscr{C} \right) &=& \sum_{\substack{j =1 \\ j\neq i}}^M \mathrm{Pr} \left( \overline{N}_{i,j} < - \frac{\sigma_{\overline{N}_{i,j}}^2}{\sigma^2} \right).
\end{IEEEeqnarray}
Since $\overline{N}_{i,j} \sim \mathcal{N}(0,\sigma_{\overline{N}_{i,j}}^2)$, from~\eqref{Eq72} it follows that
{\allowdisplaybreaks
\begin{IEEEeqnarray}{rCl} 
\gamma_i \left( \mathscr{C} \right) &=&  \sum_{\substack{j =1 \\ j\neq i}}^M \mathrm{Q} \left( \frac{\sigma_{\overline{N}_{i,j}}}{\sigma^2} \right) \\
 &=& \sum_{\substack{j =1 \\ j\neq i}}^M \mathrm{Q} \left( \sqrt{\frac{1}{2 \sigma^2} \sum_{m=1}^n (\Re(u_m(i)) - \Re(u_m(j)))^2 + (\Im(u_m(i)) - \Im(u_m(j)))^2} \right) \\
&=& \sum_{\substack{j =1 \\ j\neq i}}^M \mathrm{Q} \left( \sqrt{\frac{\sum_{m=1}^n \left| u_m(i) - u_m(j) \right|^2}{2 \sigma^2}} \right), \label{Eq70}
\end{IEEEeqnarray}
}
where, the $\mathrm{Q}$ function is in~\eqref{DefQfunc}.
For all $i \inCountK{M}$ and all $m \inCountK{n}$, $u_m(i) \in \mathcal{X} = \lbrace x^{(1)}, x^{(2)}, \ldots, x^{(L)}\rbrace$ in~\eqref{EqCIsymbols}. For $u_m(i) = x^{(\ell)}$, from~\eqref{EqNeighbor} it follows that
\begin{IEEEeqnarray}{rCl} \label{Eq76a}
\left| u_m(i) - u_m(j) \right| \leq \left| x^{(\ell)} - \bar{x}^{(\ell)} \right|
\end{IEEEeqnarray}
Since $\mathscr{C}$ is a constant composition code, from~\eqref{eq:u_measure} and Definition~\ref{DefHC}, it follows that, for all $i \inCountK{M}$, the symbol $x^{(\ell)} \in \mathcal{X}$ appears in the codeword $\boldsymbol{u}(i)$, $n P_{\mathscr{C}}(x^{(\ell)})$ number of times.
Therefore, from~\eqref{Eq70} and~\eqref{Eq76a}, it follows that,
\begin{IEEEeqnarray}{rCl} 
\gamma_i \left( \mathscr{C} \right) &\geq& \sum_{\substack{j =1 \\ j\neq i}}^M \mathrm{Q} \left( \sqrt{\frac{\sum_{\ell=1}^L n P_{\mathscr{C}}(x^{(\ell)}) \left| x^{(\ell)} - \bar{x}^{(\ell)} \right|^2}{2 \sigma^2}} \right) \\
& = & (M-1) \mathrm{Q} \left( \sqrt{\frac{\sum_{\ell=1}^L n P_{\mathscr{C}}(x^{(\ell)}) \left| x^{(\ell)} - \bar{x}^{(\ell)} \right|^2}{2 \sigma^2}} \right).
\end{IEEEeqnarray}
The average DEP $\gamma(\mathscr{C})$ in~\eqref{eq:gamma} is given by
\begin{IEEEeqnarray}{rCl}
 \gamma \left(\mathscr{C} \right)  &\geq& \frac{1}{M} \sum_{i=1}^M (M-1) \mathrm{Q} \left( \sqrt{\frac{\sum_{\ell=1}^L n P_{\mathscr{C}}(x^{(\ell)}) \left| x^{(\ell)} - \bar{x}^{(\ell)} \right|^2}{2 \sigma^2}} \right) \\
   &=& (M-1) \mathrm{Q} \left( \sqrt{\frac{\sum_{\ell=1}^L n P_{\mathscr{C}}(x^{(\ell)}) \left| x^{(\ell)} - \bar{x}^{(\ell)} \right|^2}{2 \sigma^2}} \right). \label{Eq74}
        \end{IEEEeqnarray}
    
    Finally, from~\eqref{EqGammaUpperbound} and~\eqref{Eq74}, it follows that
\begin{IEEEeqnarray}{rCl}
\epsilon & \geq & (M-1) \mathrm{Q} \left( \sqrt{\frac{\sum_{\ell=1}^L n P_{\mathscr{C}}(x^{(\ell)}) \left| x^{(\ell)} - \bar{x}^{(\ell)} \right|^2}{2 \sigma^2}} \right),
\end{IEEEeqnarray}
which completes the proof.
\end{proof}

A first upper bound on the information rate is obtained by upper bounding the number of codewords that a code might possess given the particular types $P_{\boldsymbol{u}(1)}$, $P_{\boldsymbol{u}(2)}$, $\ldots$, $P_{\boldsymbol{u}(n)}$ in~\eqref{eq:u_measure}; or the average type  $P_{\mathscr{C}}$ in~\eqref{eq:p_bar}.
The following lemma introduces such an upper bound for the case of a constant composition code.

\begin{lemma} \label{lemma:R_upperbound}
Given a constant composition $(n,M,\mathcal{X},P,\epsilon,B,\delta)$-code $\mathscr{C}$ for the random transformation in~\eqref{EqChannelModel} of the form in~\eqref{Eqnm_code}, the  information transmission rate $R(\mathscr{C})$ in~\eqref{EqR} is such that
\begin{equation}
\label{EqRbound}
R(\mathscr{C}) = \frac{1}{n} \log_2 \left( \frac{n!}{\prod_{\ell=1}^L (nP_{\mathscr{C}}(x^{(\ell)}))!}\right)  \leq \log_2 L,\end{equation}
where, $P_{\mathscr{C}}$ is the type defined in~\eqref{eq:p_bar} and $L$ is the number of channel input symbols in~\eqref{EqL}.
\end{lemma}
\begin{proof}
The largest number of codewords of length $n$ that can be formed using $L$ channel input symbols is $L^n$. In this case, the codewords do not exhibit the type $P_{\mathscr{C}}$. 
Hence, from~\eqref{EqR}, in the absence of a constraint on the type $P_{\mathscr{C}}$, the largest information transmission rate is $\log_2 L$ bits per channel use which is the inequality on the right-hand side of~\eqref{EqRbound}.

Alternatively, given a code type $P_{\mathscr{C}}$ that satisfies~\eqref{EqHomogeneousCodes}, the number of codewords that can be constructed is given by 
\begin{IEEEeqnarray}{l}
\label{Rfactorials}
M = \binom{n}{nP_{\mathscr{C}}(x^{(1)})} \binom{n-nP_{\mathscr{C}}(x^{(1)})}{nP_{\mathscr{C}}(x^{(2)})}  \ldots \binom{n-\sum_{\ell=1}^{L-1} nP_{\mathscr{C}}(x^{(\ell)})}{nP_{\mathscr{C}}(x^{(L)})}
= \frac{n!}{\prod_{\ell=1}^L (nP_{\mathscr{C}}(x^{(\ell)}))!}.
\end{IEEEeqnarray}
Therefore, the information rate $R(\mathscr{C})$ in~\eqref{EqR} satisfies
\begin{IEEEeqnarray}{rCl}
R(\mathscr{C}) &=& \frac{1}{n} \log_2 \left( \frac{n!}{\prod_{\ell=1}^L (nP_{\mathscr{C}}(x^{(\ell)}))!}\right),
\end{IEEEeqnarray}
which completes the proof.
\end{proof}
 Equality in~\eqref{EqRbound} holds when the type induced by the codewords of the code $\mathscr{C}$ is uniform, i.e., for all $x \in \mathcal{X}$ it holds that
 \begin{equation} 
    P_{\mathscr{C}}(x) = \frac{1}{L}.
\end{equation}
 Though the bound on $R$ in~\eqref{EqRbound} is tight, it can quickly become computationally infeasible as the block-length $n$ increases. 
 The following theorem provides a tractable approximation of the bound in Lemma~\ref{lemma:R_upperbound}  in terms of the entropy of the type $P_{\mathscr{C}}$. 
 
 \begin{theorem} \label{CorRUpperBoundRelax}
Given a constant composition $(n,M,\mathcal{X},P,\epsilon,B,\delta)$-code $\mathscr{C}$ for the random transformation in~\eqref{EqChannelModel} of the form in~\eqref{Eqnm_code}, the information transmission rate $R(\mathscr{C})$ in~\eqref{EqR} is such that
\begin{IEEEeqnarray}{rCl} \label{EqCorRUpperBoundRelax}
R(\mathscr{C}) &\leq& 
H\left( P_{\mathscr{C}} \right)  + \frac{1}{n^2} \left( \frac{1}{12} - \sum_{\ell=1}^L \frac{1}{12 P_{\mathscr{C}}(x^{(\ell)}) +1}\right) 
 + \frac{1}{n} \left(\log \left( \sqrt{2\pi}\right)   - \sum_{\ell=1}^L \log\sqrt{2\pi P_{\mathscr{C}}(x^{(\ell)})} \right) - \frac{\log n}{n} \left( \frac{L-1}{2} \right), \IEEEeqnarraynumspace
\end{IEEEeqnarray}
where, $P_{\mathscr{C}}$ is the type defined in~\eqref{eq:p_bar} and $L$ is the number of channel input symbols in~\eqref{EqL}.
\end{theorem}
\begin{proof} 
From~\eqref{EqRbound}, for the information transmission rate $R(\mathscr{C})$ of code $\mathscr{C}$, it holds that
\begin{IEEEeqnarray}{rCl}
\label{EqRrelax}
R(\mathscr{C}) &\leq& \frac{1}{n} \log(n!) - \frac{1}{n} \sum_{\ell=1}^L \log\left(\left(nP_{\mathscr{C}}(x^{(\ell)})\right)!\right).
\end{IEEEeqnarray}
Using the Stirling's approximation~\cite{robbins1955remark} on the factorial terms yields
\begin{IEEEeqnarray}{rCl}
  \label{EqStirlingnP}
  (nP_{\mathscr{C}}(x^{(\ell)}))!  & \geqslant &\sqrt{2\pi} \left(nP_{\mathscr{C}}(x^{(\ell)}) \right)^{nP_{\mathscr{C}}(x^{(\ell)}) +\frac{1}{2}}
\exp\left(-nP_{\mathscr{C}}(x^{(\ell)}) + \frac{1}{12nP_{\mathscr{C}}(x^{(\ell)}) +1} \right), \mbox{ and }\\
  \label{EqStirlingn}
  n !  & \leqslant &\sqrt{2\pi} n^{n +\frac{1}{2}} \exp\left(-n + \frac{1}{12 n} \right). 
\end{IEEEeqnarray}

From~\eqref{EqStirlingnP} and  ~\eqref{EqStirlingn}, it follows that,
{\allowdisplaybreaks
\begin{IEEEeqnarray}{rCl}
 \log\left( (nP_{\mathscr{C}}(x^{(\ell)}))! \right) 
& \geq &  \log \left( \sqrt{2\pi}\right) + \left(nP_{\mathscr{C}}(x^{(\ell)}) +\frac{1}{2}\right) \log (nP_{\mathscr{C}}(x^{(\ell)})) - nP_{\mathscr{C}}(x^{(\ell)}) + \frac{1}{12nP_{\mathscr{C}}(x^{(\ell)}) +1} \\
\nonumber
& =  &\log \left( \sqrt{2\pi}\right) +  nP_{\mathscr{C}}(x^{(\ell)}) \log (P_{\mathscr{C}}(x^{(\ell)})) +\frac{1}{2} \log (P_{\mathscr{C}}(x^{(\ell)})) \\
& &+  \left(nP_{\mathscr{C}}(x^{(\ell)}) +\frac{1}{2}\right) \log (n)  -nP_{\mathscr{C}}(x^{(\ell)}) + \frac{1}{12nP_{\mathscr{C}}(x^{(\ell)}) +1} \mbox{ and },\\
\nonumber
 \log\left( n! \right) & \leq &
\log \left( \sqrt{2\pi}\right) + \left(n +\frac{1}{2}\right) \log (n) -n + \frac{1}{12n}  \\
\label{EqRnfact}
& = & n  \log (n) - n + \frac{1}{12n} + \frac{1}{2}\log \left( 2\pi n\right).
 \end{IEEEeqnarray}
}
The sum in~\eqref{EqRrelax} satisfies, 
\begin{IEEEeqnarray}{rCl}
\nonumber
\sum_{\ell=1}^L \log\left(\left(nP_{\mathscr{C}}(x^{(\ell)})\right)!\right)
& \geq & L \log \left( \sqrt{2\pi}\right)  - nH\left( P_{\mathscr{C}} \right)  +\frac{1}{2}\sum_{\ell=1}^L \log (P_{\mathscr{C}}(x^{(\ell)})) + n \log (n) \\
\label{EqRrelaxA}
& &  +\frac{L}{2}\log(n) - n +\sum_{\ell=1}^L \frac{1}{12nP_{\mathscr{C}}(x^{(\ell)}) +1}.  \end{IEEEeqnarray}
Using~\eqref{EqRnfact} and~\eqref{EqRrelaxA} in~\eqref{EqRrelax} yields,
\begin{IEEEeqnarray}{rCl}
\nonumber
R(\mathscr{C}) &\leq& \log (n) - 1 + \frac{1}{12n^2} + \frac{1}{2n}\log \left( 2\pi n\right) - \frac{L}{n} \log \left( \sqrt{2\pi}\right)  + H\left( P_{\mathscr{C}} \right)  -\frac{1}{2n}\sum_{\ell=1}^L \log (P_{\mathscr{C}}(x^{(\ell)})) - \log (n)\\
& &- \frac{L}{2n}\log(n) + 1 - \frac{1}{n}\sum_{\ell=1}^L \frac{1}{12nP_{\mathscr{C}}(x^{(\ell)}) +1} \qquad\\
&\leq& H\left( P_{\mathscr{C}} \right)  + \frac{1}{n^2} \left( \frac{1}{12} - \sum_{\ell=1}^L \frac{1}{12 P_{\mathscr{C}}(x^{(\ell)}) +1}\right)  + \frac{1}{2n} \left( \log \left( 2\pi n\right)   - \sum_{\ell=1}^L \log (2\pi n P_{\mathscr{C}}(x^{(\ell)}))  \right) \\
&=& H\left( P_{\mathscr{C}} \right)  + \frac{1}{n^2} \left( \frac{1}{12} - \sum_{\ell=1}^L \frac{1}{12 P_{\mathscr{C}}(x^{(\ell)}) +1}\right) + \frac{1}{n} \left(\log \left( \sqrt{2\pi}\right)   - \sum_{\ell=1}^L \log\sqrt{2\pi P_{\mathscr{C}}(x^{(\ell)})} \right) - \frac{\log n}{n} \left( \frac{L-1}{2} \right) \label{Eq90}, \IEEEeqnarraynumspace
\end{IEEEeqnarray}
which completes the proof.
\end{proof}

Note that all terms in~\eqref{Eq90}, except the entropy $H\left( P_{\mathscr{C}} \right)$, vanish with the block-length $n$. This implies that the information rate is essentially constrained by the entropy of the channel input symbols. In particular, note that  $H\left( P_{\mathscr{C}}\right) \leqslant \log_2 L$. 
\begin{remark}
It is important to note here that, even though it may not be explicitly clear from the expressions, the bounds on the information rate and the DEP are in fact dependent on each other.
Each of the bounds in Subsection~\ref{SubsecInformationEnergyTransmission} are a function of the type $P_{\mathscr{C}}$ which defines the dependence between these bounds. 
\end{remark}
\subsection{Energy Transmission Rate and EOP} \label{SubsecEnergyConverse}

The following lemma introduces a lower bound on the EOP that holds for all $(n,M,\mathcal{X},P,\epsilon,B,\delta)$-codes. 
\begin{lemma} \label{LemmaB}
Given an $(n,M,\mathcal{X},P,\epsilon,B,\delta)$-code for the random transformation in~\eqref{EqChannelModel} of the form in~\eqref{Eqnm_code}, the following holds:
\begin{IEEEeqnarray}{rCl}
 \label{eq:B_bound_lemma}
\delta \geq \frac{1}{M} \sum_{i=1}^M \mathds{1}_{\lbrace e_i < B \rbrace},
\end{IEEEeqnarray}
where, for all $i \inCountK{M}$, the real $e_i \in [0,\infty)$ is in~\eqref{Eqei}.
\end{lemma}
\begin{proof}
The result follows from~\eqref{Eq25} and~\eqref{eq:delta}.
%
\end{proof}
The following theorem provides the bound on the EOP for the class of constant composition codes.
\begin{theorem} \label{LemmaBhomogeneous}
Given a constant composition $(n,M,\mathcal{X},P,\epsilon,B,\delta)$-code $\mathscr{C}$ for the random transformation in~\eqref{EqChannelModel} of the form in~\eqref{Eqnm_code}, the following holds:
\begin{IEEEeqnarray}{rCl} \label{EqboundEOP}
\delta = \mathds{1}_{\lbrace e_\mathscr{C} < B \rbrace},
\end{IEEEeqnarray}
where $e_\mathscr{C} \in [0,\infty)$ is in~\eqref{EqeiHomogeneous}.
\end{theorem}
\begin{proof}
From~\eqref{Eq25}, the average EOP for the $(n,M,\mathcal{X},P,\epsilon,B,\delta)$-code $\mathscr{C}$ is given by
\begin{IEEEeqnarray}{rCl}  \label{Eq40}
 \theta(\mathscr{C},B) = \frac{1}{M} \sum_{i=1}^M \mathds{1}_{\lbrace e_i < B \rbrace}.
 \end{IEEEeqnarray}
 Since $\mathscr{C}$ is a constant composition code, from~\eqref{Eq40} and~\eqref{EqeiHomogeneous} it follows that
 \begin{IEEEeqnarray}{rCl} 
 \theta(\mathscr{C},B) &=& \frac{1}{M} \sum_{i=1}^M \mathds{1}_{\lbrace e_\mathscr{C} < B \rbrace} \\
 &=& \mathds{1}_{\lbrace e_\mathscr{C} < B \rbrace}. \label{Eq42}
 \end{IEEEeqnarray}
 From~\eqref{eq:delta} and~\eqref{Eq42}, it follows that
\begin{IEEEeqnarray}{rCl} \label{Eq477}
\delta = \mathds{1}_{\lbrace e_\mathscr{C} < B \rbrace}.
\end{IEEEeqnarray} 
 This completes the proof.
\end{proof}

The following lemma provides an upper bound on the energy transmission rate $B$. 
\begin{lemma} \label{LemmaBnew}
Given an $(n,M,\mathcal{X},P,\epsilon,B,\delta)$-code $\mathscr{C}$ for the random transformation in~\eqref{EqChannelModel} of the form in~\eqref{Eqnm_code}, the energy transmission rate $B$ satisfies,
\begin{IEEEeqnarray}{rCl} \label{EqConverseB}
B \leq \bar{e}_j, \quad  \text{if } \delta \leq \frac{\sum_{k=1}^{j} y_k}{M}, \ j \in \left\lbrace 1,2,3, \ldots, M' \right\rbrace
\end{IEEEeqnarray}
where, the positive integer $M'$ is in~\eqref{EqUniqueLevels}; and, for all $j \inCountK{M'}$, $\bar{e}_j$ is in~\eqref{EqUniqueLevels} and $y_j$ is in~\eqref{EqYindicator}.
\end{lemma}
\begin{proof}
From Lemma~\ref{LemmaB}, it follows that,
\begin{IEEEeqnarray}{rCl} \label{Eq46b}
 \sum_{i=1}^M \mathds{1}_{\lbrace e_i < B \rbrace} \leq M\delta.
\end{IEEEeqnarray}
The inequality in~\eqref{Eq46b} shows that, in order to achieve an EOP less than or equal to $\delta$, the number of codewords that have energy less than $B$ (given by $ \sum_{i=1}^M \mathds{1}_{\lbrace e_i < B \rbrace}$) can at most be equal to $\lfloor M \delta \rfloor$. Using this information, the upper bound on the energy transmission rate $B$ can be calculated as follows.

From~\eqref{EqYindicator} and the definition of the EOP in~\eqref{Eq25}, it can be seen that the EOP $\delta$ can only take values from the following set:
\begin{IEEEeqnarray}{rCl} \label{Eq47b}
 \left\lbrace 0, \frac{y_1}{M}, \frac{y_1 + y_2}{M}, \ldots, \frac{\sum_{j=1}^{M'-1} y_j}{M}, 1  \right\rbrace. 
 \end{IEEEeqnarray}
 For $\delta \leq \frac{y_1}{M}$, at most $\lfloor M \delta \rfloor = y_1$ codewords can have energy less than $B$. This is possible if and only if
\begin{IEEEeqnarray}{rCl} 
B \leq \bar{e}_1.
\end{IEEEeqnarray}
 For $ \delta \leq \frac{y_1 + y_2}{M}$, at most $\lfloor M \delta \rfloor = y_1 + y_2$ codewords can have energy less than $B$ which is possible only if
\begin{IEEEeqnarray}{rCl} 
B \leq \bar{e}_2.
\end{IEEEeqnarray}
Using a similar argument for all possible values of $\delta$ in~\eqref{Eq47b},
the upper bound for $B$ is given by the following:
\begin{IEEEeqnarray}{rCl} \label{Eq42b}
B \leq \bar{e}_j, \quad  \text{if } \delta \leq \frac{\sum_{k=1}^{j} y_k}{M}, \ j \in \left\lbrace 1,2,3, \ldots, M' \right\rbrace.
\end{IEEEeqnarray}
This completes the proof.
\end{proof}
Consider the special case in which all the $M$ codewords carry a different amount of energy. That is, for all $i \inCountK{M}$, $e_1 \neq e_2 \neq \ldots \neq e_M$. Assume without loss of generality that
\begin{IEEEeqnarray}{rCl}
e_1 < e_2 < \ldots < e_M.
\end{IEEEeqnarray}
For this particular case, the upper bound on $B$ is given by the following:
\begin{IEEEeqnarray}{rCl} \label{Eq40b}
B \leq e_i, \ \ \text{if } \delta \leq \frac{i}{M}, i \inCountK{M}.
\end{IEEEeqnarray}
Note that~\eqref{Eq40b} is a special case of~\eqref{Eq42b} for $M' = M$.

In the following, a brief discussion is presented on the interpretation of the results presented in this section
\subsection{Discussion}
The bounds derived in Sections~\ref{sec:results} define the impossibility region for $(n,M,\mathcal{X},P,\epsilon,B,\delta)$-codes for the random transformation in~\eqref{EqChannelModel} of the form in~\eqref{Eqnm_code}.
These bounds reveal that trade-offs exist between the information rate $R$, the energy rate $B$, the DEP $\epsilon$ and the EOP $\delta$ of an $(n,M,\mathcal{X},P,\epsilon,B,\delta)$-code. These trade-offs are a result of the dependence of these parameters on the type $P_{\mathscr{C}}$ in~\eqref{eq:p_bar}. For instance, the upper bound on the information rate in~\eqref{EqCorRUpperBoundRelax} and the lower bound on the DEP in~\eqref{EqboundDEP} are both a function of the type $P_{\mathscr{C}}$. Though it is not made obvious from~\eqref{EqboundEOP} and~\eqref{EqConverseB}, the lower bound on the EOP and the upper bound on $B$ also depend on the type via~\eqref{EqeiHomogeneous} and~\eqref{Eq27c}. 

The information rate in~\eqref{EqCorRUpperBoundRelax} is essentially upper bounded by the entropy of $P_{\mathscr{C}}$. Hence, codes  with codewords such that every channel input symbol is used the same number of times are less constrained in terms of the information rate. This is the case in which  $P_{\mathscr{C}}$ is a uniform distribution. Alternatively, using a uniform type $P_{\mathscr{C}}$ might reduce the energy transmission rate significantly. For instance, assume that the set of channel input symbols is such that for at least one pair $(x^{(1)},x^{(2)}) \in \mathcal{X}^2$, with $\mathcal{X}$ in~\eqref{EqCIsymbols}, it holds that $\left| x^{(1)} \right| < \left| x^{(2)} \right|$. Then, from~\eqref{EqeiHomogeneousb}, it is clear that using the symbol $x^{(1)}$ equally often as $x^{(2)}$ constraints the energy $e_\mathscr{C}$ and hence, the energy rate $B$. Contrary to this, codes that exhibit the largest energy rates are those in which the symbols that have the largest magnitude are used most often. This clearly deviates from the uniform distribution and thus, constrains the information rate~$R$. 

The following section characterizes an achievable information-energy region for SIET. Given a set of channel input symbols, an achievable region defines tuples of information rate, energy rate, DEP and EOP that can be achieved by some code built using a given set of channel input symbols. 
%
\section{Achievability} \label{SecAchievability}
The first step in characterizing the achievability bounds for SIET in the finite block-length regime with finite set of channel input symbols is the construction of an $(n,M,\mathcal{X},P,\epsilon,B,\delta)$-code $\mathscr{C}$ of the form in~\eqref{Eqnm_code}. The code is constructed over a given set of channel input symbols $\mathcal{X}$. The set of channel input symbols is a modulation constellation represented by a finite subset of $\complex$. We observe that, any set of channel input symbols $\mathcal{X}$ in~\eqref{EqCIsymbols} can be described as a collection of $C \in \ints$ \emph{layers}, where, a layer is a subset of symbols in $\mathcal{X}$ that have the same magnitude. For all $c \in \lbrace 1,2, \ldots, C \rbrace$, denote by  $L_c \in \ints$ the number of symbols and by $A_c \in \reals^+$ the amplitude of the symbols in layer $c$.
 For characterizing the achievability bounds, we consider the sets of channel input symbols  in which the symbols in layer $c$ are equally spaced along the circle of radius $A_c$ (see Fig.~\ref{FigConstellation}). This is because, given the number of symbols in a layer, the symbols being equally spaced is the most favorable in terms of the DEP. This is also the structure of standard constellations employed in wireless communication (for instance QPSK, 16-QAM, 64-QAM).
Let $\alpha_c \in [0,2\pi]$ denote the phase shift of the symbols in layer $c$. The layer $c$ is denoted as
\begin{subequations}
\label{EqnmCodeCircle}
\begin{equation} \label{EqLayerCircle}
    \mathcal{U}(A_c,L_c,\alpha_c) \triangleq \left\lbrace x_c^{\left( \ell \right)} = A_c \exp\left(\mathrm{i} \left(\frac{2\pi}{L_c} \ell+\alpha_c\right)\right) \subseteq \complex: \ell \in \left \lbrace 0, 1,2, \ldots, (L_c-1)\right \rbrace \right\rbrace,
\end{equation}
where $\mathrm{i}$ is the complex unit. Using this notation, the set of channel input symbols can be described by the following set 
\begin{IEEEeqnarray}{rCl} \label{EqConstellationCircle}
    \mathcal{X} & = & \bigcup_{c=1}^C \mathcal{U}(A_c,L_c,\alpha_c).
\end{IEEEeqnarray} 
The vector of the amplitudes in~\eqref{EqConstellationCircle} is denoted by
\begin{equation} \label{EqAcVector}
\boldsymbol{A} = \left( A_1, A_2, \ldots, A_C\right)^{\sf{T}};
\end{equation}
the vector of the number of symbols in each layer in~\eqref{EqConstellationCircle} is denoted by
\begin{equation} \label{EqLcVector}
\boldsymbol{L} = \left( L_1, L_2, \ldots, L_C\right)^{\sf{T}};
\end{equation}
and the vector of the phase shifts of the symbols in each layer in~\eqref{EqConstellationCircle} is denoted by
\begin{equation} \label{EqAlphaVector}
\boldsymbol{\alpha} = \left( \alpha_1, \alpha_2, \ldots, \alpha_C\right)^{\sf{T}}.
\end{equation}
The total number of symbols $L$ in~\eqref{EqL} for $\mathcal{X}$ in~\eqref{EqConstellationCircle} is
\begin{equation} \label{EqLSum}
    L = \sum_{c=1}^C L_c.
\end{equation}

Without any loss of generality, assume that
\begin{equation} \label{EqAmplitudesOrder}
    A_1 > A_2 > \ldots > A_C.
\end{equation}
%
%

Using this representation of the set of channel input symbols, the construction of the $(n,M,\mathcal{X},P)$-code $\mathscr{C}$ is accomplished as follows.
\subsection{Code Construction}
For all $c \inCountK{C}$, let $p_c$ be the frequency with which symbols of the $c$\ts{th} layer appear in the $(n,M,\mathcal{X},P,\epsilon,B,\delta)$-code $\mathscr{C}$. The resulting probability vector is denoted by
\begin{equation} \label{Eqp}
    \boldsymbol{p} = \left( p_1,p_2, \ldots, p_C \right)^{\sf{T}}.
\end{equation}
It should be noted that the probability vector $\boldsymbol{p}$ is a free parameter in the construction of code $\mathscr{C}$. It can be fixed to satisfy the required energy transmission rate, information transmission rate, the EOP, or the DEP based on the specific requirements of the system under consideration. 

For transmitting message index $i \inCountK{M}$, the transmitter uses the codeword 
\begin{IEEEeqnarray}{rCl} \label{Eq116i}
\boldsymbol{u}(i) = \left(u_1(i), u_2(i), \ldots, u_n(i) \right) \in \mathcal{X}^n,
\end{IEEEeqnarray}
with $\mathcal{X}$ in~\eqref{EqConstellationCircle}. Using the probability vector $\boldsymbol{p}$ in~\eqref{Eqp} which determines the type $P_{\mathscr{C}}$ in~\eqref{eq:p_bar}, and the set of symbols $\mathcal{X}$ in~\eqref{EqConstellationCircle}, the set of codewords $\left\lbrace \boldsymbol{u}(1), \boldsymbol{u}(2), \ldots, \boldsymbol{u}(M) \right\rbrace$ is completely defined. 

For all $c \inCountK{C}$, it holds that,
\begin{IEEEeqnarray}{rCl} \label{Eqpc}
    p_c &=& \frac{1}{Mn} \sum_{\ell = 1}^{L_c} \sum_{i=1}^M \sum_{m=1}^n \mathds{1}_{\{u_m(i) = x_c^{(\ell)}\}} \\
    &=& \frac{1}{M} \sum_{\ell = 1}^{L_c} \sum_{i=1}^M  P_{\boldsymbol{u}(i)}\left(x_c^{(\ell)} \right), \label{Eq115i}
\end{IEEEeqnarray}
where, for all $i \inCountK{M}$ and all $m \inCountK{n}$, the complex $u_m(i)$ in~\eqref{Eq116i} is the $m$\ts{th} channel input symbol of the codeword $\boldsymbol{u}(i)$ and $x_c^{(\ell)}$ is in~\eqref{EqLayerCircle}. The equality in~\eqref{Eq115i} follows from~\eqref{eq:u_measure}.
The symbols within a layer are used with the same frequency in $\mathscr{C}$ because they are equivalent in terms of the parameters of concern (i.e., $R$, $B$, $\epsilon$, $\delta$). 
More precisely, in this construction, using two symbols within the same layer with different frequencies would have no impact on $B$, $\epsilon$, or $\delta$. On the other hand, it would negatively impact the information rate $R$.
Therefore, for all $c \inCountK{C}$ and for all $\ell \inCountK{L_c}$, the frequency with which the symbol $x_c^{(\ell)}$ appears in $\mathscr{C}$ is given by
\begin{IEEEeqnarray}{rCl}
    P_{\mathscr{C}}(x_c^{(\ell)}) &\triangleq& \frac{1}{Mn} \sum_{i=1}^M \sum_{m=1}^n \mathds{1}_{\{u_m(i)= x_c^{(\ell)}\}}, \\
\label{EqTypeCircle}    &=& \frac{p_c}{L_c}.
\end{IEEEeqnarray}
%

In this construction, a message index $i$ is said to have been decoded correctly only if all of the $n$ constituent symbols of the codeword $\boldsymbol{u}(i)$ are decoded correctly. Though decoding at the level of symbols is sub-optimal, employing this decoding strategy allows us to obtain a closed form solution for the DEP, which is in turn used for determining the number of symbols $L_c$ that can be accommodated in layer $c$ while ensuring that the DEP does not exceed $\epsilon$.

The decoding region for codeword $\boldsymbol{u}(i)$ is denoted by $\mathcal{D}_i$ and that for the symbol $u_m(i)$ is denoted by $\mathcal{D}_{i,m}$ such that the following holds:
\begin{equation} \label{EqDecodingCodeword}
    \mathcal{D}_i = \mathcal{D}_{i,1} \times \mathcal{D}_{i,2} \times \ldots \times \mathcal{D}_{i,n}, 
\end{equation}
where, for all $i \inCountK{M}$, $c \inCountK{C}$, $m \inCountK{n}$, and $\ell \inCountK{L_c}$, when $u_m(i) = x_c^{(\ell)}$ in~\eqref{EqLayerCircle}, then, $\mathcal{D}_{i,m} = \mathcal{G}_c^{(\ell)}$.

The decoding set $\mathcal{G}_c^{(\ell)} \subseteq \complex$ associated with symbol $x_c^{(\ell)}$ is a circle of radius $r_c \in \reals^+$ centered at $x_c^{(\ell)}$. That is,
\begin{IEEEeqnarray}{rCl}
\mathcal{G}_c^{(\ell)} &=& \left\lbrace y \in \complex : \left| y - x_c^{\left(\ell \right)}\right|^2 \leq r_c^2 \right\rbrace  \label{EqDecodingCircleComplex}  \\
&=&\left\lbrace y \in \complex : \left(\Re(y) - \Re(x_c^{\left(\ell \right)})\right)^2 + \left(\Im(y) - \Im(x_c^{\left(\ell \right)})\right)^2 \leq r_c^2 \right\rbrace.  \label{EqDecodingCircle}
\end{IEEEeqnarray}
To ensure that the decoding regions are mutually disjoint, for all $c \inCountK{C}$, the amplitudes $A_c$ in~\eqref{EqLayerCircle} and the radii in~\eqref{EqDecodingCircleComplex} are chosen to satisfy the following:
\begin{IEEEeqnarray}{rCl} 
\label{EqAmplitudesDifference}
A_{c+1} - A_{c} \geq r_{c+1} + r_{c}.
\end{IEEEeqnarray}

The vector of the radii in~\eqref{EqDecodingCircleComplex} is denoted by
\begin{equation} \label{EqRadiusVector}
\boldsymbol{r} = \left( r_1, r_2, \ldots, r_C\right)^{\sf{T}}.
\end{equation}

This defines a family of $(n,M,\mathcal{X},P)$-codes denoted by
\begin{IEEEeqnarray}{rCl}
\label{EqFamily}
{\sf C} \left(C,\boldsymbol{A},\boldsymbol{L},\boldsymbol{\alpha},\boldsymbol{p},\boldsymbol{r} \right),
\end{IEEEeqnarray}
\end{subequations}
with the number of layers $C$ in~\eqref{EqConstellationCircle}, $\boldsymbol{A}$ in~\eqref{EqAcVector}, $\boldsymbol{L}$ in~\eqref{EqLcVector}, $\boldsymbol{\alpha}$ in~\eqref{EqAlphaVector} $\boldsymbol{p}$ in~\eqref{Eqp} and $\boldsymbol{r}$ in~\eqref{EqRadiusVector}. The following subsections characterize the achievable region for codes in the family ${\sf C} \left(C,\boldsymbol{A},\boldsymbol{L},\boldsymbol{\alpha},\boldsymbol{p},\boldsymbol{r} \right)$.

\subsection{Achievability Bounds} 
The first set of results presented in Subsection~\ref{SecAchievableBounds} provide a lower bound on the achievable DEP (Lemma~\ref{LemmaRadius}) and an upper bound on the achievable information rate (Lemma~\ref{LemmaAchievableR}) for codes in the family ${\sf C} \left(C,\boldsymbol{A},\boldsymbol{L},\boldsymbol{\alpha},\boldsymbol{p},\boldsymbol{r} \right)$. In the process, we also obtain a lower bound on the radii of decoding sets $\boldsymbol{r}$ in~\eqref{EqRadiusVector} (Lemma~\ref{CorollaryRadius}) and
 an upper bound on the number of symbols $L_c$ in each layer $c$ of the constructed code (Lemma~\ref{LemmaLc}). The second set of results in Subsection~\ref{SecAchievableBounds2} provide a lower bound on the achievable EOP (Lemmas~\ref{LemmaAchievableEOP}) and an upper bound on the energy transmission rate (Lemma~\ref{LemmaBAchievable}) for codes in the family ${\sf C} \left(C,\boldsymbol{A},\boldsymbol{L},\boldsymbol{\alpha},\boldsymbol{p},\boldsymbol{r} \right)$. Lemmas~\ref{LemmaAlphaImmaterial} and~\ref{LemmaAlphaImmaterialStrong} in Subsection~\ref{SubsecMisc} prove that the DEP and EOP achieved by the constructed family of codes are independent of rotation of the set of channel input symbols $\mathcal{X}$ and of the phase shift vector $\boldsymbol{\alpha}$ in~\eqref{EqAlphaVector}, respectively.

These bounds together define the achievable region of the codes of the form in~\eqref{EqnmCodeCircle} from the family ${\sf C} \left(C,\boldsymbol{A},\boldsymbol{L},\boldsymbol{\alpha},\boldsymbol{p},\boldsymbol{r} \right)$ in~\eqref{EqFamily} as given by the following theorem:
\begin{theorem} \label{TheoremAchievableRegion}
An $\left( n,M,\mathcal{X},P \right)$-code $\mathscr{C}$ of the form in~\eqref{EqnmCodeCircle} from the family ${\sf C} \left(C,\boldsymbol{A},\boldsymbol{L},\boldsymbol{\alpha},\boldsymbol{p},\boldsymbol{r} \right)$ in~\eqref{EqFamily} is an $\left(n,M,\mathcal{X},P,\epsilon,B,\delta \right)$-code if the number of layers $C$, and for all $c \inCountK{C}$, the parameters $A_c$, $L_c$, $r_c$, and $p_c$ satisfy the following:
\begin{subequations}
\begin{IEEEeqnarray}{l}
 \epsilon \geq 1 - \frac{1}{M}\sum_{i = 1}^M\prod_{c=1}^{C}  \left( 1-\exp \left(-\frac{r_c^2}{\sigma^2}\right) \right)^{n \sum_{\ell=1}^{L_c} P_{\boldsymbol{u}(i)}(x_c^{(\ell)})}, \label{EqTheoremrc} \\
L_c \leq \left\lfloor \frac{\pi}{2\arcsin{\frac{r_c}{2A_c}}} \right\rfloor, \label{EqTheoremLc}\\
R\left(\mathscr{C} \right) \leq \log_2 \sum_{c=1}^C \left\lfloor \frac{\pi}{2\arcsin{\frac{r_c}{2A_c}}} \right\rfloor, \label{EqTheoremR} \\
 \delta \geq  \frac{1}{M} \sum_{i=1}^M \mathds{1}_{\left\lbrace \left( k_1 \sum_{c=1}^C \sum_{\ell=1}^{L_c} n P_{\boldsymbol{u}(i)} \left( x_c^{(\ell)} \right) A_c^2 + k_2 \sum_{c=1}^C \sum_{\ell=1}^{L_c} n P_{\boldsymbol{u}(i)} \left( x_c^{(\ell)} \right) A_c^4 \right) < B \right\rbrace}, \label{EqTheoremB}
\end{IEEEeqnarray}
\end{subequations}
where, $R\left(\mathscr{C} \right)$ is the information rate in~\eqref{EqR}; $k_1$ and $k_2$ are non-zero positive constants in~\eqref{Eqei}; $P_{\boldsymbol{u}(i)}$ is the type in~\eqref{eq:u_measure}, and the real $\sigma^2$ is defined in~\eqref{EqDensities}.
\end{theorem}
\begin{proof}
The proof will follow from Lemmas~\ref{LemmaRadius},~\ref{LemmaLc},~\ref{LemmaAchievableRnew},~\ref{LemmaAchievableEOP} in the following subsections.
\end{proof}
What follows in Subsections~\ref{SecAchievableBounds} and~\ref{SecAchievableBounds2} below is a set of lemmas that prove Theorem~\ref{TheoremAchievableRegion}.
\subsection{Information Transmission Rate and DEP}  \label{SecAchievableBounds}
The following lemma provides a lower bound on the DEP of codes from the family ${\sf C} \left(C,\boldsymbol{A},\boldsymbol{L},\boldsymbol{\alpha},\boldsymbol{p},\boldsymbol{r} \right)$ in~\eqref{EqFamily}.
\begin{lemma} \label{LemmaRadius}
Consider an $(n,M,\mathcal{X},P,\epsilon,B,\delta)$-code $\mathscr{C}$ of the form in~\eqref{EqnmCodeCircle} from the family ${\sf C} \left(C,\boldsymbol{A},\boldsymbol{L},\boldsymbol{\alpha},\boldsymbol{p},\boldsymbol{r} \right)$ in~\eqref{EqFamily}. The parameters $r_1, r_2, \ldots, r_C$ in~\eqref{EqDecodingCircle} satisfy the following:
\begin{equation}
\epsilon \geq 1 - \frac{1}{M}\sum_{i = 1}^M\prod_{c=1}^{C}  \left( 1-\exp \left(-\frac{r_c^2}{\sigma^2} \right) \right)^{n \sum_{\ell=1}^{L_c} P_{\boldsymbol{u}(i)}(x_c^{(\ell)})},
\end{equation}
where, the type  $P_{\boldsymbol{u}(i)}$ is defined in~\eqref{eq:u_measure}, the real $\sigma^2$ is defined in~\eqref{EqDensities}, and $x_c^{(\ell)} \in \mathcal{U}(A_c,L_c,\alpha_c)$, with $\mathcal{U}(A_c,L_c,\alpha_c)$ in \eqref{EqLayerCircle}.
\end{lemma}
\begin{proof}
From~\eqref{EqDEPi} and~\eqref{eq:gamma}, the average DEP of the $\left(n,M,\mathcal{X},P,\epsilon,B,\delta \right)$-code $\mathscr{C}$ is given by
{\allowdisplaybreaks
\begin{IEEEeqnarray}{rCl}
    \gamma(\mathscr{C}) &=& 1 - \frac{1}{M} \sum_{i=1}^M \int_{\mathcal{D}_i} f_{\boldsymbol{Y}|\boldsymbol{X}}(\boldsymbol{y}|\boldsymbol{u}(i)) \mathrm{d}\boldsymbol{y} \\
    &=& 1 - \frac{1}{M} \sum_{i=1}^M \int_{\mathcal{D}_{i,1} \times \mathcal{D}_{i,2} \times \ldots \times \mathcal{D}_{i,n}} \prod_{m=1}^n f_{Y|X} \left( y|u_m(i) \right) \mathrm{d}y \label{Eq138b} \\
&=& 1 - \frac{1}{M} \sum_{i=1}^M \prod_{m=1}^n \int_{\mathcal{D}_{i,m}} f_{Y|X} \left( y|u_m(i) \right) \mathrm{d}y \label{Eq139} \\ 
&=& 1 - \frac{1}{M} \sum_{i=1}^M \prod_{c=1}^C \prod_{\ell=1}^{L_c} \left(  \int_{\mathcal{G}_c^{(\ell)}} f_{Y|X}(y|x_c^{(\ell)}) \mathrm{d}y  \right)^{n P_{\boldsymbol{u}(i)}(x_c^{(\ell)})}, \label{Eq140}
\end{IEEEeqnarray}
}
where, the equality in~\eqref{Eq138b} follows due to~\eqref{EqYXdistribution} and~\eqref{EqDecodingCodeword}, and~\eqref{Eq139} follows due to Fubini's theorem.
Using~\eqref{EqDensities} in~\eqref{Eq140} yields,
\begin{equation} \label{EqAvgGamma}
    \gamma \left(\mathscr{C} \right) = 1 - \frac{1}{M} \sum_{i=1}^M \prod_{c=1}^C \prod_{\ell=1}^{L_c} \left(  \int_{\mathcal{G}_c^{(\ell)}} \frac{1}{\pi \sigma^2} \exp \left(- \frac{(\Re(y)-\Re(x_c^{(\ell)}))^2 +(\Im(y) - \Im(x_c^{(\ell)}))^2}{\sigma^2} \right) \mathrm{d}y  \right)^{n P_{\boldsymbol{u}(i)}(x_c^{(\ell)})}.
\end{equation}
Evaluating the integral term in~\eqref{EqAvgGamma} for all $\ell \inCountK{L_c}$ yields,
\begin{IEEEeqnarray}{rcl}
&&\int_{\mathcal{G}_c^{(\ell)}} \frac{1}{\pi \sigma^2} \exp \left(- \frac{(\Re(y)-\Re(x_c^{(\ell)}))^2 +(\Im(y) - \Im(x_c^{(\ell)}))^2}{\sigma^2} \right) \mathrm{d}y \nonumber \\
= &&\int_{\Im(x_c^{(\ell)})-r_c}^{\Im(x_c^{(\ell)})+r_c} \int_{\Re(x_c^{(\ell)})-\sqrt{r_c^2-(v-\Im(x_c^{(\ell)}))^2}}^{\Re(x_c^{(\ell)})+\sqrt{r_c^2-(v-\Im(x_c^{(\ell)}))^2}} \frac{1}{\pi \sigma^2} \exp \left(-\frac{(u-\Re(x_c^{(\ell)}))^2+(v-\Im(x_c^{(\ell)}))^2}{\sigma^2} \right) \mathrm{d}u \mathrm{d}v,\\
= &&\int_{-r_c}^{r_c} \int_{-\sqrt{r_c^2-v^2}}^{\sqrt{r_c^2-v^2}} \frac{1}{\pi \sigma^2} \exp \left(-\frac{u^2+v^2}{\sigma^2} \right) \mathrm{d}u \mathrm{d}v, \label{Eq82}\\
= &&\int_{0}^{\pi} \int_{0}^{\frac{r_c}{\sigma}} \frac{1}{2\pi} \exp \left(-\zeta^2 \right) \zeta \mathrm{d}\zeta \mathrm{d}\eta, \label{Eq83} \\
= &&\left( 1-\exp \left(-\frac{r_c^2}{\sigma^2}\right) \right). \label{Eq94}
\end{IEEEeqnarray}
The equality in~\eqref{Eq83} is obtained from the change of variables $u = \sigma \zeta \cos \eta, v = \sigma \zeta \sin \eta$. Plugging~\eqref{Eq94} in~\eqref{EqAvgGamma} yields,
\begin{IEEEeqnarray}{rCl}
\gamma \left(\mathscr{C} \right) &=& 1 - \frac{1}{M} \sum_{i=1}^M \prod_{c=1}^C \prod_{\ell=1}^{L_c} \left( 1-\exp \left(-\frac{r_c^2}{\sigma^2}\right) \right)^{n P_{\boldsymbol{u}(i)}(x_c^{(\ell)})} ,\\
 \label{Eq95}    &=& 1 - \frac{1}{M}\sum_{i = 1}^M \prod_{c=1}^C \left( 1-\exp \left(-\frac{r_c^2}{\sigma^2} \right) \right)^{n \sum_{\ell=1}^{L_c} P_{\boldsymbol{u}(i)}(x_c^{(\ell)})}.
\end{IEEEeqnarray}
From~\eqref{EqGammaUpperbound}, for $\mathscr{C}$ to be an $\left( n,M,\mathcal{X},P,\epsilon \right)$-code, the following must hold:
\begin{equation} \label{EqGammaSumBound}
    \gamma \left(\mathscr{C} \right) \leq \epsilon.
\end{equation}
This implies that,
\begin{IEEEeqnarray}{rCl}
\epsilon \geq 1 - \frac{1}{M}\sum_{i = 1}^M \prod_{c=1}^C \left( 1-\exp \left(-\frac{r_c^2}{\sigma^2} \right) \right)^{n \sum_{\ell=1}^{L_c} P_{\boldsymbol{u}(i)}(x_c^{(\ell)})},
    \end{IEEEeqnarray}
which completes the proof.
\end{proof}
From Definition~\ref{DefHC} and~\eqref{EqTypeCircle}, it follows that, for a constant composition code $\mathscr{C}$ of the form in~\eqref{EqnmCodeCircle} from the family ${\sf C} \left(C,\boldsymbol{A},\boldsymbol{L},\boldsymbol{\alpha},\boldsymbol{p},\boldsymbol{r} \right)$ in~\eqref{EqFamily}, it holds that,
\begin{IEEEeqnarray}{rCl}
    \label{Eq100}
   P_{\boldsymbol{u}(i)}(x_c^{(\ell)}) = P_{\mathscr{C}}(x_c^{(\ell)}) = \frac{p_c}{L_c}.
\end{IEEEeqnarray} 

The following result for constant composition codes follows from~\eqref{Eq100} and Lemma~\ref{LemmaRadius}.
%
\begin{corollary}
\label{CorRadiusHomogeneous}
Consider a constant composition $(n,M,\mathcal{X},P,\epsilon,B,\delta)$-code $\mathscr{C}$ of the form in~\eqref{EqnmCodeCircle} from the family ${\sf C} \left(C,\boldsymbol{A},\boldsymbol{L},\boldsymbol{\alpha},\boldsymbol{p},\boldsymbol{r} \right)$ in~\eqref{EqFamily}. The parameters $r_1, r_2, \ldots, r_C$ in~\eqref{EqDecodingCircle} satisfy the following:
\begin{equation}
 \epsilon \geq 1- \prod_{c=1}^{C}  \left( 1-\exp \left(-\frac{r_c^2}{\sigma^2} \right) \right)^{n p_c},
\end{equation}
where, the real $\sigma^2$ is defined in~\eqref{EqDensities}, and $x_c^{(\ell)} \in \mathcal{U}(A_c,L_c,\alpha_c)$, with $\mathcal{U}(A_c,L_c,\alpha_c)$ in \eqref{EqLayerCircle}.
\end{corollary}
\begin{remark} \label{Remark1}
It is essential to note a couple of things about the choice of decoding sets in~\eqref{EqDecodingCircle} and the derived bounds on the achievable DEP.
Firstly, the circular form of decoding sets in~\eqref{EqDecodingCircle} is a choice made in order to obtain closed form tractable expressions for the lower bounds on the DEP as in Lemma~\ref{LemmaRadius} and Corollary~\ref{CorRadiusHomogeneous}.
This is enabled by the inherent circular symmetry of the complex AWGN noise.
The bounds in Lemma~\ref{LemmaRadius} and Corollary~\ref{CorRadiusHomogeneous} reveal the dependence of the DEP on the various parameters of the code such as the type. This is extremely instructive for studying the trade-offs between the parameters of the constructed family of codes.
Secondly, this choice of decoding sets provides guidelines for constructing the set of channel input symbols $\mathcal{X}$ while ensuring that the DEP does not exceed a required value which is controlled using the choice of the radii of the decoding sets in~\eqref{EqRadiusVector}.
Thirdly, even though the decoding sets defined in~\eqref{EqDecodingCircle} are used to determine the achievable lower bounds on the DEP, the achievable DEP can in fact be reduced further by employing the optimal MAP decoder.
This is because, the circular decoding sets in~\eqref{EqDecodingCircle} would be subsets of the decoding sets obtained using a MAP decoder.
The circular decoding sets could be used only for the purposes of code construction and to ensure that the DEP requirement is respected while the decoding process could make use of the optimal MAP decoder.
\end{remark}
For the special case where for all $c \inCountK{C}$, $r_c = r$ in~\eqref{EqRadiusVector}, the following lemma provides a lower bound on the radius $r$.
%
\begin{lemma} \label{CorollaryRadius}
Consider an $(n,M,\mathcal{X},P,\epsilon,B,\delta)$-code $\mathscr{C}$ of the form in~\eqref{EqnmCodeCircle} from the family ${\sf C} \left(C,\boldsymbol{A},\boldsymbol{L},\boldsymbol{\alpha},\boldsymbol{p},\boldsymbol{r} \right)$ in~\eqref{EqFamily} such that, for all $c \inCountK{C}$, $r_c = r$ in~\eqref{EqDecodingCircle}. Then, the parameter $r$ satisfies:
\begin{equation}\label{EqRminEqual}
    r \geq \sqrt{\sigma^2 \log \left( \frac{1}{1-(1-\epsilon)^{\frac{1}{n}}} \right) },
\end{equation}
where, the real $\sigma^2$ is defined in~\eqref{EqDensities}.
\end{lemma}
\begin{proof}
If the parameters $r_c$ in~\eqref{EqDecodingCircle} are such that, for all $c \inCountK{C}$, $r_c = r$, the average DEP in~\eqref{Eq95} is given by:
\begin{IEEEeqnarray}{rCl}
\gamma \left( \mathscr{C} \right) &=& 1 - \frac{1}{M}\sum_{i = 1}^M \prod_{c=1}^C \left( 1-e^{-\frac{r^2}{\sigma^2}} \right)^{n \sum_{\ell=1}^{L_c} P_{\boldsymbol{u}(i)}(x_c^{(\ell)})}, \\
&=& 1 - \frac{1}{M}\sum_{i = 1}^M \left( 1-e^{-\frac{r^2}{\sigma^2}} \right)^{n \sum_{c=1}^C \sum_{\ell=1}^{L_c} P_{\boldsymbol{u}(i)}(x_c^{(\ell)})}, \\
&=& 1 - \frac{1}{M}\sum_{i = 1}^M \left( 1-e^{-\frac{r^2}{\sigma^2}} \right)^{n}, \\
&=& 1 - \left( 1-e^{-\frac{r^2}{\sigma^2}} \right)^{n}. \label{Eq105}
\end{IEEEeqnarray}
From~\eqref{EqGammaUpperbound}, for $\mathscr{C}$ to be an $\left( n,M,\mathcal{X},P,\epsilon \right)$-code, the following must hold:
\begin{IEEEeqnarray}{rCl}
\epsilon \geq 1 - \left( 1-e^{-\frac{r^2}{\sigma^2}} \right)^{n}.
    \end{IEEEeqnarray}
This implies that
\begin{IEEEeqnarray}{rCl}
    r \geq \sqrt{\sigma^2 \log \left( \frac{1}{1-(1-\epsilon)^{\frac{1}{n}}} \right) }.
\end{IEEEeqnarray}
This completes the proof.
\end{proof}
%
%
%
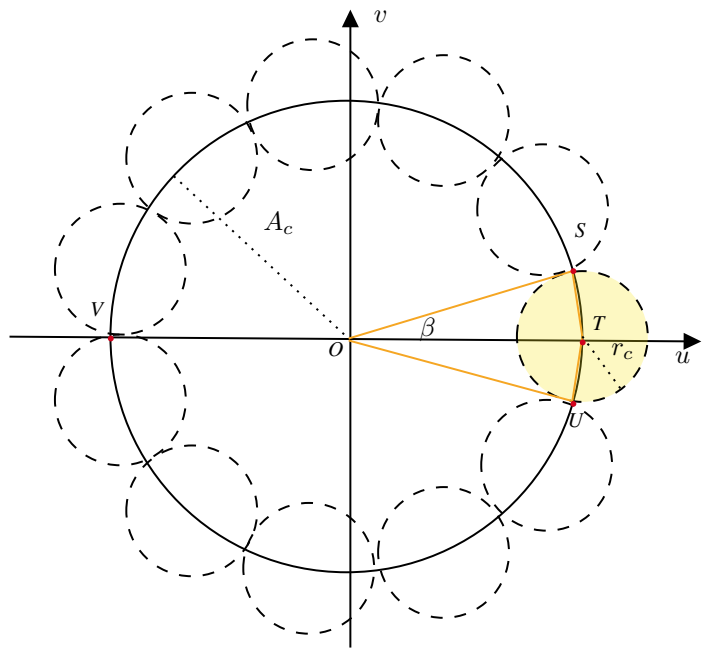
\begin{figure}
    \centering
\tikzset{every picture/.style={line width=0.75pt}} 

\begin{tikzpicture}[x=0.75pt,y=0.75pt,yscale=-1,xscale=1]

\draw  [color={rgb, 255:red, 0; green, 0; blue, 0 }  ,draw opacity=1 ] (189.67,167) .. controls (189.67,101.28) and (242.94,48) .. (308.67,48) .. controls (374.39,48) and (427.67,101.28) .. (427.67,167) .. controls (427.67,232.72) and (374.39,286) .. (308.67,286) .. controls (242.94,286) and (189.67,232.72) .. (189.67,167) -- cycle ;
\draw    (310.67,324) -- (310.67,5) ;
\draw [shift={(310.67,2)}, rotate = 450] [fill={rgb, 255:red, 0; green, 0; blue, 0 }  ][line width=0.08]  [draw opacity=0] (8.93,-4.29) -- (0,0) -- (8.93,4.29) -- cycle    ;
\draw    (138.77,167.08) -- (484.67,169.48) ;
\draw [shift={(487.67,169.5)}, rotate = 180.4] [fill={rgb, 255:red, 0; green, 0; blue, 0 }  ][line width=0.08]  [draw opacity=0] (8.93,-4.29) -- (0,0) -- (8.93,4.29) -- cycle    ;
\draw  [fill={rgb, 255:red, 248; green, 231; blue, 28 }  ,fill opacity=0.27 ][dash pattern={on 4.5pt off 4.5pt}] (394.67,167) .. controls (394.67,148.77) and (409.44,134) .. (427.67,134) .. controls (445.89,134) and (460.67,148.77) .. (460.67,167) .. controls (460.67,185.23) and (445.89,200) .. (427.67,200) .. controls (409.44,200) and (394.67,185.23) .. (394.67,167) -- cycle ;
\draw [color={rgb, 255:red, 245; green, 166; blue, 35 }  ,draw opacity=1 ]   (309.8,168) -- (422.61,134) ;
\draw [color={rgb, 255:red, 245; green, 166; blue, 35 }  ,draw opacity=1 ]   (422.61,134) -- (427.67,167) ;
\draw [color={rgb, 255:red, 245; green, 166; blue, 35 }  ,draw opacity=1 ]   (309.8,169) -- (422.61,199.6) ;
\draw [color={rgb, 255:red, 245; green, 166; blue, 35 }  ,draw opacity=1 ]   (427.67,167) -- (422.61,199.6) ;
\draw  [dash pattern={on 0.84pt off 2.51pt}]  (221.63,86.08) -- (308.67,167) ;
\draw  [dash pattern={on 0.84pt off 2.51pt}]  (427.67,167) -- (447.13,193.58) ;
\draw  [dash pattern={on 4.5pt off 4.5pt}] (197.67,255) .. controls (197.67,236.77) and (212.44,222) .. (230.67,222) .. controls (248.89,222) and (263.67,236.77) .. (263.67,255) .. controls (263.67,273.23) and (248.89,288) .. (230.67,288) .. controls (212.44,288) and (197.67,273.23) .. (197.67,255) -- cycle ;
\draw  [dash pattern={on 4.5pt off 4.5pt}] (161.67,199) .. controls (161.67,180.77) and (176.44,166) .. (194.67,166) .. controls (212.89,166) and (227.67,180.77) .. (227.67,199) .. controls (227.67,217.23) and (212.89,232) .. (194.67,232) .. controls (176.44,232) and (161.67,217.23) .. (161.67,199) -- cycle ;
\draw  [dash pattern={on 4.5pt off 4.5pt}] (324.67,58) .. controls (324.67,39.77) and (339.44,25) .. (357.67,25) .. controls (375.89,25) and (390.67,39.77) .. (390.67,58) .. controls (390.67,76.23) and (375.89,91) .. (357.67,91) .. controls (339.44,91) and (324.67,76.23) .. (324.67,58) -- cycle ;
\draw  [dash pattern={on 4.5pt off 4.5pt}] (374.67,103) .. controls (374.67,84.77) and (389.44,70) .. (407.67,70) .. controls (425.89,70) and (440.67,84.77) .. (440.67,103) .. controls (440.67,121.23) and (425.89,136) .. (407.67,136) .. controls (389.44,136) and (374.67,121.23) .. (374.67,103) -- cycle ;
\draw  [dash pattern={on 4.5pt off 4.5pt}] (161.67,133) .. controls (161.67,114.77) and (176.44,100) .. (194.67,100) .. controls (212.89,100) and (227.67,114.77) .. (227.67,133) .. controls (227.67,151.23) and (212.89,166) .. (194.67,166) .. controls (176.44,166) and (161.67,151.23) .. (161.67,133) -- cycle ;
\draw  [dash pattern={on 4.5pt off 4.5pt}] (197.67,77) .. controls (197.67,58.77) and (212.44,44) .. (230.67,44) .. controls (248.89,44) and (263.67,58.77) .. (263.67,77) .. controls (263.67,95.23) and (248.89,110) .. (230.67,110) .. controls (212.44,110) and (197.67,95.23) .. (197.67,77) -- cycle ;
\draw  [dash pattern={on 4.5pt off 4.5pt}] (258.67,50) .. controls (258.67,31.77) and (273.44,17) .. (291.67,17) .. controls (309.89,17) and (324.67,31.77) .. (324.67,50) .. controls (324.67,68.23) and (309.89,83) .. (291.67,83) .. controls (273.44,83) and (258.67,68.23) .. (258.67,50) -- cycle ;
\draw  [dash pattern={on 4.5pt off 4.5pt}] (256.67,284) .. controls (256.67,265.77) and (271.44,251) .. (289.67,251) .. controls (307.89,251) and (322.67,265.77) .. (322.67,284) .. controls (322.67,302.23) and (307.89,317) .. (289.67,317) .. controls (271.44,317) and (256.67,302.23) .. (256.67,284) -- cycle ;
\draw  [dash pattern={on 4.5pt off 4.5pt}] (376.67,232) .. controls (376.67,213.77) and (391.44,199) .. (409.67,199) .. controls (427.89,199) and (442.67,213.77) .. (442.67,232) .. controls (442.67,250.23) and (427.89,265) .. (409.67,265) .. controls (391.44,265) and (376.67,250.23) .. (376.67,232) -- cycle ;
\draw  [dash pattern={on 4.5pt off 4.5pt}] (324.67,276) .. controls (324.67,257.77) and (339.44,243) .. (357.67,243) .. controls (375.89,243) and (390.67,257.77) .. (390.67,276) .. controls (390.67,294.23) and (375.89,309) .. (357.67,309) .. controls (339.44,309) and (324.67,294.23) .. (324.67,276) -- cycle ;

\draw (422.42,108.03) node [anchor=north west][inner sep=0.75pt]  [font=\large] [align=left] {{\fontfamily{ptm}\selectfont {\footnotesize \textit{S}}}};
\draw (298,168) node [anchor=north west][inner sep=0.75pt]  [font=\large] [align=left] {{\fontfamily{ptm}\selectfont {\footnotesize \textit{O}}}};
\draw (431.35,155) node [anchor=north west][inner sep=0.75pt]  [font=\large] [align=left] {{\fontfamily{ptm}\selectfont {\footnotesize \textit{T}}}};
\draw (418.82,204.23) node [anchor=north west][inner sep=0.75pt]  [font=\large] [align=left] {{\fontfamily{ptm}\selectfont {\footnotesize \textit{U}}}};
\draw (265.29,102.71) node [anchor=north west][inner sep=0.75pt]  [font=\normalsize] [align=left] {$\displaystyle A_c$};
\draw (440.79,169.21) node [anchor=north west][inner sep=0.75pt]  [font=\normalsize] [align=left] {$\displaystyle r_c$};
\draw (473,173) node [anchor=north west][inner sep=0.75pt]   [align=left] {$\displaystyle u$};
\draw (321,1) node [anchor=north west][inner sep=0.75pt]   [align=left] {$\displaystyle v$};
\draw (344,156) node [anchor=north west][inner sep=0.75pt]   [align=left] {$\displaystyle \beta $};

\draw (178,147.9) node [anchor=north west][inner sep=0.75pt]  [font=\large] [align=left] {{\fontfamily{ptm}\selectfont {\footnotesize \textit{V}}}};
\draw (185,165) node [anchor=north west][inner sep=0.75pt]  [font=\huge] [align=left] {{\fontfamily{ptm}\selectfont \textcolor[rgb]{0.82,0.01,0.11}{.}}};

\draw (418.31,131) node [anchor=north west][inner sep=0.75pt]  [font=\huge] [align=left] {{\fontfamily{ptm}\selectfont \textcolor[rgb]{0.82,0.01,0.11}{.}}};
\draw (423,167) node [anchor=north west][inner sep=0.75pt]  [font=\huge] [align=left] {{\fontfamily{ptm}\selectfont \textcolor[rgb]{0.82,0.01,0.11}{.}}};
\draw (418.31,198) node [anchor=north west][inner sep=0.75pt]  [font=\huge] [align=left] {{\fontfamily{ptm}\selectfont \textcolor[rgb]{0.82,0.01,0.11}{.}}};

\end{tikzpicture}

    \caption{Graphical representation of the symbols in layer $c$ defined in~\eqref{EqLayerCircle}}
    \label{FigConstellation}
\end{figure}
%
%
The information rate achievable by a code is a function of the number of channel input symbols $L$ in~\eqref{EqLSum} which in turn is a function of the number of symbols in each layer of the set of channel inputs $\mathcal{X}$ in~\eqref{EqConstellationCircle}. The following lemma provides an upper bound on the number of symbols in a layer $c \inCountK{C}$ denoted by $L_c$.
\begin{lemma} \label{LemmaLc}
Consider an $(n,M,\mathcal{X},P,\epsilon,B,\delta)$-code $\mathscr{C}$ of the form in~\eqref{EqnmCodeCircle} from the family ${\sf C} \left(C,\boldsymbol{A},\boldsymbol{L},\boldsymbol{\alpha},\boldsymbol{p},\boldsymbol{r} \right)$ in~\eqref{EqFamily}. Then, for all $c \inCountK{C}$, the number of symbols in layer $c$ of $\mathcal{X}$ must satisfy the following:
\begin{equation}
    L_c \leq \left\lfloor \frac{\pi}{2\arcsin{\frac{r_c}{2A_c}}} \right\rfloor,
\end{equation}
where, $r_c$ is the radius of the decoding regions $\mathcal{G}_c^{(1)}, \mathcal{G}_c^{(2)}, \ldots, \mathcal{G}_c^{(L_c)}$ in~\eqref{EqDecodingCircle} and $A_c$ is the amplitude in~\eqref{EqLayerCircle}.
\end{lemma}
\begin{proof}
For an $(n,M,\mathcal{X},P,\epsilon,B,\delta)$-code $\mathscr{C}$ in the family ${\sf C} \left(C,\boldsymbol{A},\boldsymbol{L},\boldsymbol{\alpha},\boldsymbol{p},\boldsymbol{r} \right)$, for all $c \inCountK{C}$, the radius $r_c$ of the decoding regions in~\eqref{EqRadiusVector} and the amplitude $A_c$ in~\eqref{EqAcVector} determine the number of symbols $L_c$ that can be accommodated in the layer $c$. 
\par
Layer $c$ of the form in~\eqref{EqLayerCircle} is illustrated in Fig.~\ref{FigConstellation}. Symbols in layer $c$ are distributed uniformly along the circle of radius $A_c$ centered at the origin $O$.
The maximum number of symbols in layer $c$ is equal to the number of non-overlapping circles of radius $r_c$ corresponding to the decoding regions defined in~\eqref{EqDecodingCircle} that can be placed along the circumference of the circle of radius $A_c$.
From Fig.~\ref{FigConstellation}, a circle of radius $r_c$ centered at a symbol in layer $c$ subtends angle $\angle \mbox{SOU} = \beta$ at $O$. Therefore, the maximum number of symbols $L_c$ that can be accommodated along the circle of radius $A_c$ is given by
    \begin{equation} \label{Eq109}
        L_c \leq \left\lfloor  \frac{2\pi}{\beta} \right\rfloor.
    \end{equation}
\par
%
The following lemma determines the value of the angle $\beta$ in Fig.~\ref{FigConstellation}. 
\begin{lemma} \label{LemmaGeometry}
Consider a circle of radius $A_c$ centered at $O$ and a circle of radius $r_c$ centered on the circumference of the first circle as shown in Fig.~\ref{FigConstellation}. The two circles intersect at points $S$ and $U$. Then, the angle $\angle \mbox{SOU}$ is given by the following:
\begin{equation} \label{Eq110}
 \angle \mbox{SOU} = \beta = 4\arcsin{\frac{r_c}{2A_c}}.
\end{equation}
\end{lemma}
\begin{proof}
The proof is given in Appendix~\ref{AppA}.
\end{proof}
Substituting the value of $\beta$ from~\eqref{Eq110} in~\eqref{Eq109}, it follows that the number of symbols in layer $c$ of $\mathcal{X}$ is at most
    \begin{equation}
        L_c \leq \left\lfloor \frac{\pi}{2\arcsin{\frac{r_c}{2A_c}}} \right\rfloor.
    \end{equation}
This completes the proof.
\end{proof}
Using Lemma~\ref{LemmaLc}, the achievable bound on the information transmission rate $R\left(\mathscr{C} \right)$ for code $\mathscr{C}$ from the family ${\sf C} \left(C,\boldsymbol{A},\boldsymbol{L},\boldsymbol{\alpha},\boldsymbol{p},\boldsymbol{r} \right)$ in~\eqref{EqFamily} can now be calculated. The following lemma provides this result.
\begin{lemma} \label{LemmaAchievableRnew}
Given an $(n,M,\mathcal{X},P,\epsilon,B,\delta)$-code $\mathscr{C}$ of the form in~\eqref{EqnmCodeCircle} from the family ${\sf C} \left(C,\boldsymbol{A},\boldsymbol{L},\boldsymbol{\alpha},\boldsymbol{p},\boldsymbol{r} \right)$ in~\eqref{EqFamily}, the information transmission rate $R\left(\mathscr{C} \right)$ satisfies the following:
\begin{IEEEeqnarray}{rCl}
    R\left(\mathscr{C} \right) \leq \log_2 \sum_{c=1}^C \left\lfloor \frac{\pi}{2\arcsin{\frac{r_c}{2A_c}}} \right\rfloor,
\end{IEEEeqnarray}
where, for all $c \inCountK{C}$, the radius $r_c$ is in~\eqref{EqDecodingCircle}, and the amplitude $A_c$ is in~\eqref{EqLayerCircle}.
\end{lemma}
\begin{proof}
The largest number of codewords of length $n$ that can be formed with $L$ different channel input symbols is $L^n$. Hence, from~\eqref{EqR}, it follows that
\begin{IEEEeqnarray}{rCl}
    R\left(\mathscr{C} \right) &=& \frac{\log_2 M}{n} \\
    &\leq& \frac{\log L^n}{n} \\
    &=& \log L \label{EqRStep1} \\
    &\leq& \log \sum_{c=1}^C \left\lfloor \frac{\pi}{2\arcsin{\frac{r_c}{2A_c}}} \right\rfloor, \label{EqRStep2}
\end{IEEEeqnarray}
where, the inequality in~\eqref{EqRStep2} follows from Lemma~\ref{LemmaLc} and~\eqref{EqLSum}. This completes the proof.
\end{proof}
The following lemma provides a bound on the information transmission rate for a constant composition $(n,M,\mathcal{X},P,\epsilon,B,\delta)$-code.
\begin{lemma} \label{LemmaAchievableR}
For a constant composition $(n,M,\mathcal{X},P,\epsilon,B,\delta)$-code $\mathscr{C}$ of the form in~\eqref{EqnmCodeCircle} from the family ${\sf C} \left(C,\boldsymbol{A},\boldsymbol{L},\boldsymbol{\alpha},\boldsymbol{p},\boldsymbol{r} \right)$ in~\eqref{EqFamily}, the information transmission rate $R\left(\mathscr{C} \right)$ satisfies the following:
\begin{IEEEeqnarray}{rCl} \label{Eq174b}
    R\left(\mathscr{C} \right) = \frac{1}{n} \log_2 \left( \frac{n!}{\prod_{c=1}^C \left( \left( n\frac{ p_c}{L_c} \right)! \right)^{L_c}} \right).
\end{IEEEeqnarray}
\end{lemma}
\begin{proof}
%
For a constant composition $(n,M,\mathcal{X},P,\epsilon,B,\delta)$-code $\mathscr{C}$ from the family ${\sf C} \left(C,\boldsymbol{A},\boldsymbol{L},\boldsymbol{\alpha},\boldsymbol{p},\boldsymbol{r} \right)$ for which the type $P_{\mathscr{C}}$ satisfies~\eqref{EqHomogeneousCodes}, the number of codewords that can be constructed  is given by 
\begin{IEEEeqnarray}{l}
M = \binom{n}{nP_{\mathscr{C}}(x_1^{(1)})} \binom{n-nP_{\mathscr{C}}(x_1^{(1)})}{nP_{\mathscr{C}}(x_1^{(2)})}  \ldots \binom{n-\sum_{\ell=1}^{L_1-1} nP_{\mathscr{C}}(x_1^{(\ell)})}{nP_{\mathscr{C}}(x_1^{(L_1)})} \times 
\binom{n-\sum_{\ell=1}^{L_1} nP_{\mathscr{C}}(x_1^{(\ell)})}{nP_{\mathscr{C}}(x_2^{(1)})} \nonumber \\
\binom{n-\sum_{\ell=1}^{L_1} nP_{\mathscr{C}}(x_1^{(\ell)}) - nP_{\mathscr{C}}(x_2^{(1)})}{nP_{\mathscr{C}}(x_2^{(2)})} \ldots \binom{n-\sum_{\ell=1}^{L_1} nP_{\mathscr{C}}(x_1^{(\ell)})-\sum_{\ell=1}^{L_2-1} nP_{\mathscr{C}}(x_2^{(\ell)})}{nP_{\mathscr{C}}(x_2^{(L_2)})}  \times \ldots \times 
\binom{n-\sum_{c=1}^{C-1} \sum_{\ell=1}^{L_c} nP_{\mathscr{C}}(x_c^{(\ell)})}{nP_{\mathscr{C}}(x_C^{(1)})} \nonumber \\
 \binom{n-\sum_{c=1}^{C-1} \sum_{\ell=1}^{L_c} nP_{\mathscr{C}}(x_c^{(\ell)}) - nP_{\mathscr{C}}(x_C^{(1)})}{nP_{\mathscr{C}}(x_C^{(2)})} \ldots \binom{n-\sum_{c=1}^{C-1} \sum_{\ell=1}^{L_c} nP_{\mathscr{C}}(x_c^{(\ell)})-\sum_{\ell=1}^{L_C-1} nP_{\mathscr{C}}(x_C^{(\ell)})}{nP_{\mathscr{C}}(x_C^{(L_C)})}  \\
 = \binom{n}{n\frac{p_1}{L_1}} \binom{n-n\frac{p_1}{L_1}}{n\frac{p_1}{L_1}}  \ldots \binom{n-\sum_{\ell=1}^{L_1-1} n\frac{p_1}{L_1}}{n\frac{p_1}{L_1}} \times 
\binom{n-np_1}{n\frac{p_2}{L_2}}
\binom{n-np_1 - n\frac{p_2}{L_2}}{n\frac{p_2}{L_2}} \ldots \binom{n-np_1-\sum_{\ell=1}^{L_2-1} n\frac{p_2}{L_2}}{n\frac{p_2}{L_2}}  \times \ldots \times \nonumber \\
\binom{n-\sum_{c=1}^{C-1} np_c}{n\frac{p_C}{L_C}}
 \binom{n-\sum_{c=1}^{C-1} np_c - n\frac{p_C}{L_C}}{n\frac{p_C}{L_C}} \ldots \binom{n-\sum_{c=1}^{C-1} np_c-\sum_{\ell=1}^{L_C-1} n\frac{p_C}{L_C}}{n\frac{p_C}{L_C}}  \\
 = \frac{n!}{\prod_{c=1}^C \left( \left( n\frac{ p_c}{L_c} \right)! \right)^{L_c}}
\end{IEEEeqnarray}
Therefore, the information transmission rate $R\left(\mathscr{C} \right)$ is given by
\begin{IEEEeqnarray}{rCl} \label{Eq159b}
    R\left(\mathscr{C} \right) = \frac{\log_2 M}{n} = \frac{1}{n} \log_2 \left( \frac{n!}{\prod_{c=1}^C \left( \left( n\frac{ p_c}{L_c} \right)! \right)^{L_c}} \right).
\end{IEEEeqnarray}
This completes the proof.
\end{proof}
The information transmission rate $R\left(\mathscr{C} \right)$ of a code $\mathscr{C}$ is dictated by the number of channel input symbols $L$ in~\eqref{EqLSum} and the probability vector $\boldsymbol{p}$ in~\eqref{Eqp} that defines the empirical frequency with which each symbol appears in $\mathscr{C}$. The following lemma provides the value of $\boldsymbol{p}$ that achieves the upper bound on $R\left(\mathscr{C} \right)$ for a given set of channel input symbols.
\begin{lemma} \label{LemmaRateMax}
Consider an $(n,M,\mathcal{X},P,\epsilon,B,\delta)$-code $\mathscr{C}'$ of the form in~\eqref{EqnmCodeCircle} from the family ${\sf C} \left(C,\boldsymbol{A},\boldsymbol{L},\boldsymbol{\alpha},\boldsymbol{p},\boldsymbol{r} \right)$ in~\eqref{EqFamily} with $\boldsymbol{p} = \left( p_1,p_2, \ldots, p_C \right)^{\sf{T}}$ in~\eqref{Eqp} such that, for all $c \inCountK{C}$,
\begin{equation} \label{EqpcRmax}
    p_c = \frac{L_c}{L},
\end{equation}
where, $L_c$ and $L$ are as defined in~\eqref{EqLayerCircle} and~\eqref{EqLSum} respectively.
Then, given any other $(n,M,\mathcal{X},P,\epsilon,B,\delta)$-code $\mathscr{C}$ in ${\sf C} \left(C,\boldsymbol{A},\boldsymbol{L},\boldsymbol{\alpha},\boldsymbol{p},\boldsymbol{r} \right)$ that is identical to $\mathscr{C}'$ except for the probability distribution $\boldsymbol{p}$, it holds that,
\begin{equation}
    R\left(\mathscr{C}' \right) \geq R\left(\mathscr{C} \right),
\end{equation}
where, $R\left(\mathscr{C}' \right)$ and $R\left(\mathscr{C} \right)$ are the information transmission rates for $\mathscr{C}'$ and $\mathscr{C}$, respectively.
\end{lemma}
\begin{proof}
The proof is given in Appendix~\ref{AppD}.
\end{proof}

\subsection{Energy Transmission Rate and EOP} \label{SecAchievableBounds2}
The following lemma provides the achievable bound on the EOP $\delta$ in~\eqref{eq:delta} for codes from the family ${\sf C} \left(C,\boldsymbol{A},\boldsymbol{L},\boldsymbol{\alpha},\boldsymbol{p},\boldsymbol{r} \right)$ in~\eqref{EqFamily}.
\begin{lemma} \label{LemmaAchievableEOP}
For an $(n,M,\mathcal{X},P,\epsilon,B,\delta)$-code $\mathscr{C}$ of the form in~\eqref{EqnmCodeCircle} from the family ${\sf C} \left(C,\boldsymbol{A},\boldsymbol{L},\boldsymbol{\alpha},\boldsymbol{p},\boldsymbol{r} \right)$ in~\eqref{EqFamily}, the following holds:
\begin{IEEEeqnarray}{rCl} \label{Eq213}
 \delta \geq  \frac{1}{M} \sum_{i=1}^M \mathds{1}_{\left\lbrace \left( k_1 \sum_{c=1}^C \sum_{\ell=1}^{L_c} n P_{\boldsymbol{u}(i)} \left( x_c^{(\ell)} \right) A_c^2 + k_2 \sum_{c=1}^C \sum_{\ell=1}^{L_c} n P_{\boldsymbol{u}(i)} \left( x_c^{(\ell)} \right) A_c^4 \right) < B \right\rbrace},
 \end{IEEEeqnarray}
where, $k_1$ and $k_2$ are positive real constants defined in~\eqref{Eq27c} and $P_{\boldsymbol{u}(i)}$ is the type in~\eqref{eq:u_measure}.
\end{lemma}
\begin{proof}
From~\eqref{EqLayerCircle}, for all $c \inCountK{C}$ and all $\ell \inCountK{L_c}$, the symbols $x_c^{\left( \ell \right)} \in \mathcal{U}(A_c,L_c,\alpha_c)$ in~\eqref{EqLayerCircle} are given by
\begin{IEEEeqnarray}{rCl} 
x_c^{\left( \ell \right)} = A_c \exp\left(\mathrm{i} \left(\frac{2\pi}{L_c} \ell+\alpha_c\right)\right).
\end{IEEEeqnarray}
This implies that
\begin{IEEEeqnarray}{rCl} \label{Eq218}
\left| x_c^{\left( \ell \right)} \right| = A_c.
\end{IEEEeqnarray}
From~\eqref{Eq25} and~\eqref{Eq27c}, the EOP for the $(n,M,\mathcal{X},P,\epsilon,B,\delta)$-code $\mathscr{C}$ is given by:
\begin{IEEEeqnarray}{rCl} 
 \theta(\mathscr{C},B) & = & \frac{1}{M} \sum_{i=1}^M \mathds{1}_{\left\lbrace \left( k_1 \sum_{c=1}^C \sum_{\ell=1}^{L_c} n P_{\boldsymbol{u}(i)} \left( x_c^{(\ell)} \right) \left| x_c^{(\ell)} \right|^2 + k_2 \sum_{c=1}^C \sum_{\ell=1}^{L_c} n P_{\boldsymbol{u}(i)} \left( x_c^{(\ell)} \right) \left| x_c^{(\ell)} \right|^4 \right) < B \right\rbrace} \\
& = & \frac{1}{M} \sum_{i=1}^M \mathds{1}_{\left\lbrace \left( k_1 \sum_{c=1}^C \sum_{\ell=1}^{L_c} n P_{\boldsymbol{u}(i)} \left( x_c^{(\ell)} \right) A_c^2 + k_2 \sum_{c=1}^C \sum_{\ell=1}^{L_c} n P_{\boldsymbol{u}(i)} \left( x_c^{(\ell)} \right) A_c^4 \right) < B \right\rbrace}. \label{Eq209}
 \end{IEEEeqnarray}
From~\eqref{eq:delta} and~\eqref{Eq209}, the code $\mathscr{C}$ is an $\left( n,M,\mathcal{X},P,\epsilon,B,\delta \right)$-code if the following holds:
\begin{IEEEeqnarray}{rCl} 
  \frac{1}{M} \sum_{i=1}^M \mathds{1}_{\left\lbrace \left( k_1 \sum_{c=1}^C \sum_{\ell=1}^{L_c} n P_{\boldsymbol{u}(i)} \left( x_c^{(\ell)} \right) A_c^2 + k_2 \sum_{c=1}^C \sum_{\ell=1}^{L_c} n P_{\boldsymbol{u}(i)} \left( x_c^{(\ell)} \right) A_c^4 \right) < B \right\rbrace} \leq \delta.
\end{IEEEeqnarray}
This completes the proof.
\end{proof}
The achievable bound on the EOP $\delta$ in~\eqref{eq:delta} for a constant composition code $\mathscr{C}$ from the family ${\sf C} \left(C,\boldsymbol{A},\boldsymbol{L},\boldsymbol{\alpha},\boldsymbol{p},\boldsymbol{r} \right)$ in~\eqref{EqFamily} is given by the following lemma:
\begin{lemma} \label{LemmaAmplitudeC}
For a constant composition $(n,M,\mathcal{X},P,\epsilon,B,\delta)$-code $\mathscr{C}$ of the form in~\eqref{EqnmCodeCircle} from the family ${\sf C} \left(C,\boldsymbol{A},\boldsymbol{L},\boldsymbol{\alpha},\boldsymbol{p},\boldsymbol{r} \right)$ in~\eqref{EqFamily}, for all $c \inCountK{C}$, the parameters $p_c$ in~\eqref{Eqpc} satisfy the following:
\begin{IEEEeqnarray}{rCl}
 \delta = \mathds{1}_{\left\lbrace \left( k_1 \sum_{c=1}^C n p_c A_c^2 + k_2 \sum_{c=1}^C n p_c A_c^4 \right) < B \right\rbrace},
 \end{IEEEeqnarray}
where, $k_1$ and $k_2$ are positive real constants defined in~\eqref{Eq27c}.
\end{lemma}
\begin{proof}
From~\eqref{Eq25} and~\eqref{EqeiHomogeneous}, the EOP for the constant composition code $\mathscr{C}$ is given by:
\begin{IEEEeqnarray}{rCl} 
 \theta(\mathscr{C},B) & = & \frac{1}{M} \sum_{i=1}^M \mathds{1}_{\left\lbrace \left( k_1 \sum_{c=1}^C \sum_{\ell=1}^{L_c} n P_{\mathscr{C}} \left( x_c^{(\ell)} \right) \left| x_c^{(\ell)} \right|^2 + k_2 \sum_{c=1}^C \sum_{\ell=1}^{L_c} n P_{\mathscr{C}} \left( x_c^{(\ell)} \right) \left| x_c^{(\ell)} \right|^4 \right) < B \right\rbrace} \\
& = & \mathds{1}_{\left\lbrace \left( k_1 \sum_{c=1}^C \sum_{\ell=1}^{L_c} n P_{\mathscr{C}} \left( x_c^{(\ell)} \right) A_c^2 + k_2 \sum_{c=1}^C \sum_{\ell=1}^{L_c} n P_{\mathscr{C}} \left( x_c^{(\ell)} \right) A_c^4 \right) < B \right\rbrace} \label{Eq191}\\
& = & \mathds{1}_{\left\lbrace \left( k_1 \sum_{c=1}^C \sum_{\ell=1}^{L_c} n \frac{p_c}{L_c} A_c^2 + k_2 \sum_{c=1}^C \sum_{\ell=1}^{L_c} n \frac{p_c}{L_c} A_c^4 \right) < B \right\rbrace} \label{Eq192} \\
& = & \mathds{1}_{\left\lbrace \left( k_1 \sum_{c=1}^C n p_c A_c^2 + k_2 \sum_{c=1}^C n p_c A_c^4 \right) < B \right\rbrace}. \label{Eq266}
 \end{IEEEeqnarray}
 The equality in~\eqref{Eq191} follows from~\eqref{Eq218} and~\eqref{Eq192} follows from~\eqref{EqTypeCircle}.
From~\eqref{eq:delta} and~\eqref{Eq266}, it follows that the code $\mathscr{C}$ is an $\left( n,M,\mathcal{X},P,\epsilon,B,\delta \right)$-code if the following holds:
\begin{IEEEeqnarray}{rCl} 
 \delta = \mathds{1}_{\left\lbrace \left( k_1 \sum_{c=1}^C n p_c A_c^2 + k_2 \sum_{c=1}^C n p_c A_c^4 \right) < B \right\rbrace}.
\end{IEEEeqnarray}
This completes the proof.
\end{proof}
The following lemma provides an upper bound on the energy transmission rate $B$ for $(n,M,\mathcal{X},P,\epsilon,B,\delta)$-codes of the form in~\eqref{EqnmCodeCircle} from the family ${\sf C} \left(C,\boldsymbol{A},\boldsymbol{L},\boldsymbol{\alpha},\boldsymbol{p},\boldsymbol{r} \right)$.
\begin{lemma} \label{LemmaBAchievable}
For an $(n,M,\mathcal{X},P,\epsilon,B,\delta)$-code $\mathscr{C}$ of the form in~\eqref{EqnmCodeCircle} from the family ${\sf C} \left(C,\boldsymbol{A},\boldsymbol{L},\boldsymbol{\alpha},\boldsymbol{p},\boldsymbol{r} \right)$ in~\eqref{EqFamily}, the energy transmission rate $B$ satisfies the following:
\begin{IEEEeqnarray}{rCl} \label{EqBAchievable}
B \leq \bar{e}_j, \quad  \text{if } \delta \leq \frac{\sum_{k=1}^{j} y_k}{M}, \ j \in \left\lbrace 1,2,3, \ldots, M' \right\rbrace
\end{IEEEeqnarray}
where, the positive integer $M'$ is in~\eqref{EqUniqueLevels}; and for all $j \inCountK{M'}$ $\bar{e}_j$ is in~\eqref{EqUniqueLevels} and $y_j$ is in~\eqref{EqYindicator}.
\end{lemma}
\begin{proof}
The proof follows on the same lines as that for Lemma~\ref{LemmaBnew} where, for all $i \inCountK{M}$, the energy $e_i$ in~\eqref{Eqei} is given by
\begin{IEEEeqnarray}{rCl}
e_i = \sum_{c=1}^C \sum_{\ell=1}^{L_c} n P_{\boldsymbol{u}(i)} \left( x_c^{(\ell)} \right) A_c^2 + k_2 \sum_{c=1}^C \sum_{\ell=1}^{L_c} n P_{\boldsymbol{u}(i)} \left( x_c^{(\ell)} \right) A_c^4.
\end{IEEEeqnarray}
\end{proof}
The following lemma provides the value of $\boldsymbol{p}$ that achieves the lower bound on the EOP $\theta\left(\mathscr{C}, B \right)$ for code $\mathscr{C}$ from the family ${\sf C} \left(C,\boldsymbol{A},\boldsymbol{L},\boldsymbol{\alpha},\boldsymbol{p},\boldsymbol{r} \right)$ in~\eqref{EqFamily}.
\begin{lemma} \label{LemmaMaxEnergy}
Consider an $(n,M,\mathcal{X},P,\epsilon,B,\delta)$-code $\mathscr{C}$ of the form in~\eqref{EqnmCodeCircle} from the family ${\sf C} \left(C,\boldsymbol{A},\boldsymbol{L},\boldsymbol{\alpha},\boldsymbol{p},\boldsymbol{r} \right)$ in~\eqref{EqFamily} with $\boldsymbol{p} = \left( p_1,p_2, \ldots, p_C \right)^{\sf{T}}$ in~\eqref{Eqp} such that, 
\begin{equation} \label{EqpcMaxEnergy}
  p_c =
  \begin{cases}
    1 & \text{if } c=1,
    \\
    0 & \text{otherwise }.
  \end{cases}
\end{equation}
Then, given any other $(n,M,\mathcal{X},P,\epsilon,B,\delta)$-code $\mathscr{C}'$ in the family ${\sf C} \left(C,\boldsymbol{A},\boldsymbol{L},\boldsymbol{\alpha},\boldsymbol{p},\boldsymbol{r} \right)$ that is identical to $\mathscr{C}$ except for the probability distribution $\boldsymbol{p}$, it holds that
\begin{equation}
    \theta\left(\mathscr{C}', B \right) \geq \theta\left(\mathscr{C}, B \right),
\end{equation}
where $\theta$ is the average EOP in~\eqref{eq:theta_def}.
\end{lemma}
\begin{proof}
The proof is given in Appendix~\ref{AppE}.
\end{proof}
\begin{remark}
Lemmas~\ref{LemmaRateMax} and~\ref{LemmaMaxEnergy} provide the optimal choice of the type $\boldsymbol{p}$ for the information transmission rate $R$ and the energy transmission rate $B$, respectively.
It should be noted that these choices are optimal only for these specific quantities when no other parameters of the code are taken into consideration.
We observe that, the choice of $\boldsymbol{p}$ provided by the two lemmas is very different. 
Lemma~\ref{LemmaRateMax} says that, for better information transmission rates, the frequency of usage of symbols within a layer $c$ should be proportional to the number of symbols in that layer $L_c$.
This translates into using all symbols within the set of channel input symbols $\mathcal{X}$ uniformly within the code.
Contrary to this, Lemma~\ref{LemmaMaxEnergy} says that, for better energy transmission rates, it is best to only use the symbols in the first layer of $\mathcal{X}$ which has the most energy. 
Thus, choosing $\boldsymbol{p}$ to maximize the information transmission rate will typically result in lower energy transmission rates and vice versa.
 \end{remark}
 \subsection{Miscellaneous Results} \label{SubsecMisc}
 Given the number of layers $C$ in~\eqref{EqConstellationCircle}, the vector of amplitudes $\boldsymbol{A}$ in~\eqref{EqAcVector}, the vector of the number of symbols in each layer $\boldsymbol{L}$ in~\eqref{EqLcVector}, the vector of phase shifts $\boldsymbol{\alpha}$ in~\eqref{EqAlphaVector},  the probability vector $\boldsymbol{p}$ in~\eqref{Eqp} and the radius vector $\boldsymbol{r}$ in~\eqref{EqRadiusVector}, infinitely many sets of channel input symbols $\mathcal{X}$ in~\eqref{EqConstellationCircle} can be generated by rotating $\mathcal{X}$. Rotating $\mathcal{X}$ implies changing the phase shift of all the symbols in $\mathcal{X}$ by the same angle. It is easy to see that the information transmission rate $R\left(\mathscr{C} \right)$ of the code $\mathscr{C}$ is independent of such rotations since it only depends on the number of symbols $L$ in~\eqref{EqLSum} and the probability vector $\boldsymbol{p}$. However, the impact of such rotations on the DEP $ \gamma \left(\mathscr{C} \right)$ and the EOP $\theta \left(\mathscr{C}, B \right)$ of $\mathscr{C}$ is not immediately obvious. Lemmas~\ref{LemmaAlphaImmaterial} and~\ref{LemmaAlphaImmaterialStrong} prove that the DEP and the EOP also remain unchanged for any rotation of the set $\mathcal{X}$.
\begin{lemma} \label{LemmaAlphaImmaterial} 
Consider two $(n,M,\mathcal{X},P,\epsilon,B,\delta)$-codes $\mathscr{C}$ and $\mathscr{C}'$, both of the form in~\eqref{EqnmCodeCircle} from the family ${\sf C} \left(C,\boldsymbol{A},\boldsymbol{L},\boldsymbol{\alpha},\boldsymbol{p},\boldsymbol{r} \right)$ in~\eqref{EqFamily}. For all $c \inCountK{C}$ and all $\ell \inCountK{L_c}$, denote the symbols in layer $c$ in~\eqref{EqLayerCircle} of the code $\mathscr{C}$  by $x_c^{\left( \ell \right)}$ and those of $\mathscr{C}'$ by $\hat{x}_c^{\left( \ell \right)}$. For all $c \inCountK{C}$, all $\ell \inCountK{L_c}$, and $\omega \in [0,2\pi]$, it holds that,
\begin{equation} \label{EqShiftedSymbols}
    \hat{x}_c^{\left( \ell \right)} = e^{\mathrm{i} \omega} x_c^{\left( \ell \right)}.
\end{equation}
All other parameters of the codes $\mathscr{C}$ and $\mathscr{C}'$ are identical.
Then, the average DEP $\gamma$ in~\eqref{eq:gamma} and the average EOP $\theta$ in~\eqref{eq:theta_def} of $\mathscr{C}$ and $\mathscr{C}'$ satisfy the following:
\begin{IEEEeqnarray}{rCl}
 \label{Eq123}   \gamma \left(\mathscr{C} \right) &=& \gamma \left(\mathscr{C}' \right), \mbox{and} \\
 \label{Eq124}   \theta \left(\mathscr{C}, B \right) &=& \theta \left(\mathscr{C}', B \right).
\end{IEEEeqnarray}
\end{lemma}
\begin{proof} 
The proof is given in Appendix~\ref{AppB}.
\end{proof}
Lemma~\ref{LemmaAlphaImmaterial} proves that the DEP and the EOP of a code $\mathscr{C}$ are independent of the rotations of the underlying set of channel input symbols $\mathcal{X}$. In fact, \textit{for the given construction} of codes in this section, the DEP and the EOP are actually independent of the vector of phase shifts $\boldsymbol{\alpha}$ in~\eqref{EqAlphaVector} altogether. The following lemma proves this result.
\begin{lemma} \label{LemmaAlphaImmaterialStrong}
Consider $(n,M,\mathcal{X},P,\epsilon,B,\delta)$-codes $\mathscr{C}$ and $\mathscr{C}'$ of the form in~\eqref{EqnmCodeCircle} from the family ${\sf C} \left(C,\boldsymbol{A},\boldsymbol{L},\boldsymbol{\alpha},\boldsymbol{p},\boldsymbol{r} \right)$ in~\eqref{EqFamily}. The codes 
$\mathscr{C}$ and $\mathscr{C}'$ are identical except for the phase shift $\alpha_c$ in~\eqref{EqLayerCircle}. More specifically, the set of channel input symbols of $\mathscr{C}$ and $\mathscr{C}'$ are $\mathcal{X}$ and $\mathcal{X}'$, respectively such that
\begin{IEEEeqnarray}{rCl}
\mathcal{X} &=& \bigcup_{c=1}^C \mathcal{U}(A_c,L_c,\alpha_c), \ \mbox{and} \\
\mathcal{X}' &=& \bigcup_{c=1}^C \mathcal{U}(A_c,L_c,\alpha_c').
\end{IEEEeqnarray}
Then, the average DEP $\gamma$ in~\eqref{eq:gamma} and the average EOP $\theta$ in~\eqref{eq:theta_def} of $\mathscr{C}$ and $\mathscr{C}'$ satisfy the following:
\begin{IEEEeqnarray}{rCl}
 \label{Eq173} \theta \left(\mathscr{C}, B \right) &=& \theta \left(\mathscr{C}', B \right), \ \mbox{and} \\
  \label{Eq159} \gamma \left(\mathscr{C} \right) &=& \gamma \left(\mathscr{C}' \right).
\end{IEEEeqnarray}
\end{lemma}
\begin{proof} 
The proof is given in Appendix~\ref{AppC}.
\end{proof}
The following section presents a discussion of the impossibility and achievability bounds derived in this work.
\section{Discussion} \label{SecDiscussion} 
\subsubsection{Trade-offs} 
The impossibility and achievable bounds presented in this work reveal several interesting insights into the trade-offs between the information transmission rate $R$, the energy transmission rate $B$, the DEP $\epsilon$ and the EOP $\delta$. For instance, the impossibility as well as the achievable lower bounds on $\delta$ in~\eqref{eq:B_bound_lemma}  and~\eqref{Eq213}, respectively, increase as $B$ increases and vice versa. The consequence of this relationship is that a smaller $\delta$ can be achieved at the cost of lower $B$. Similarly, for increasing the impossibility or achievable upper bounds on $B$ in~\eqref{EqConverseB} and~\eqref{EqBAchievable}, respectively, a larger value of $\delta$ has to be tolerated.
 These bounds also reveal that both $B$ and $\delta$ can be improved by a code that has higher values of $e_i$ in~\eqref{Eqei}. This is achieved by using the symbols with greater energy more frequently in the code. 

The impossibility lower bound on the DEP $\epsilon$ in~\eqref{EqboundDEP} decreases as the distance between the symbols in the set of channel inputs $\mathcal{X}$ in~\eqref{EqConstellationCircle} increases. This implies that, a lower DEP can be achieved by increasing the distances between the channel input symbols. However, with the peak-amplitude constraint in place, increasing the distance between symbols implies that the number of symbols $L$ decreases. This in turn, decreases the upper bound on the information rate $R$ in~\eqref{EqRbound}. The achievable lower bound on $\epsilon$ in~\eqref{EqTheoremrc} decreases as a function of the radii $r_c$ of the decoding regions in~\eqref{EqDecodingCircle}. On the other hand, the upper bound on the achievable information rate $R$ in~\eqref{Eq174b} increases as the radii $r_c$ decrease.
These trade-offs between $R$, $B$, $\epsilon$, and $\delta$ are further illustrated using the following example.
\begin{figure}[t]
  \centering
  \includegraphics[width=0.7\textwidth]{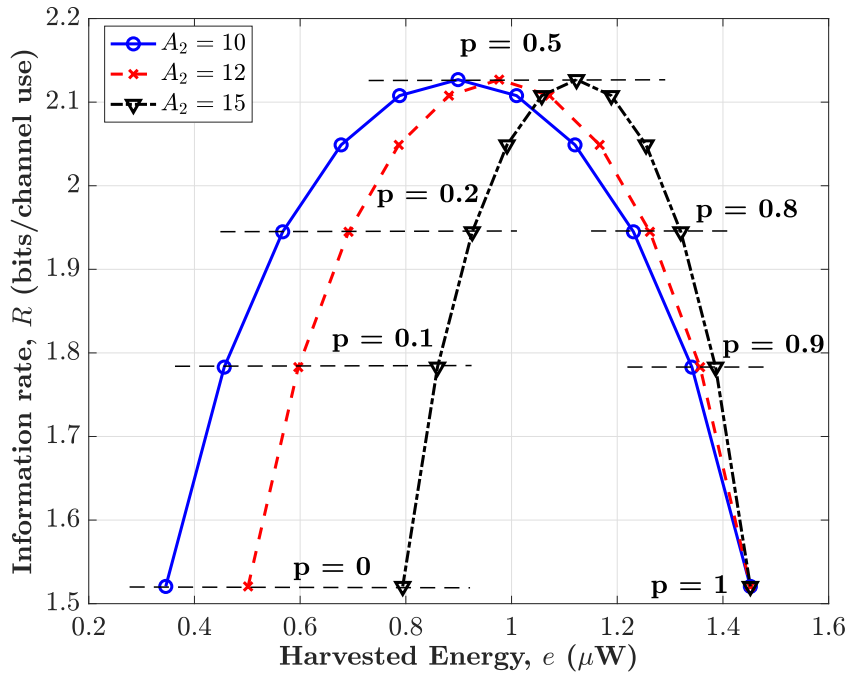}
\caption{Impossibility bounds on the information transmission rate $R$ in~\eqref{EqRbound} as a function of the harvested energy $e$ in~\eqref{Eq33}.}
\label{FigBR}
\end{figure}
%
%
%

Consider a family of constant composition $(n,M,\mathcal{X},P,\epsilon,B,\delta)$-codes $\mathscr{C}$ in the family ${\sf C} \left(C,\boldsymbol{A},\boldsymbol{L},\boldsymbol{\alpha},\boldsymbol{p},\boldsymbol{r} \right)$ in~\eqref{EqFamily} with the peak-amplitude constraint $P = 20$ millivolts in~\eqref{EqPCriteria}. The set of channel input symbols $\mathcal{X}$ in \eqref{EqConstellationCircle} is composed of two layers with $5$ symbols in each layer, \ie, $C=2$ and $L_1 = L_2 = 5$. The radius of the first layer is $A_1 = P$ and the radius of the second layer $A_2$ is varied to illustrate the trade-offs between the various parameters. The frequency with which symbols from the first layer appear in the code is $p_1 = p = 1-p_2$. That is, the vector $\boldsymbol{p}$ in~\eqref{Eqp} is given by
\begin{IEEEeqnarray}{l} \label{Eq71}
\boldsymbol{p} = \left( p, (1-p)\right)^{\sf{T}}.
 \end{IEEEeqnarray}
 The duration of the transmission in channel uses is $n = 100$. Since $\mathscr{C}$ is a constant composition code, from~\eqref{EqeiHomogeneous}, it holds that, for all $i \inCountK{M}$, 
 \begin{IEEEeqnarray}{rCl}
 e_i = e, \label{Eq33}
 \end{IEEEeqnarray}
where $e \in [0,\infty)$ is calculated as in~\eqref{EqeiHomogeneous}. 
 
Fig.~\ref{FigBR} shows the trade-offs between the impossibility bound on the information transmission rate $R$~\eqref{EqRbound} and the harvested energy $e$ in microwatts ($\mu W$) in~\eqref{Eq33} as a function of $p$ in~\eqref{Eq71}. Each curve in the figure is generated for some value of $A_2 < A_1$ by varying the value of $p \in [0,1]$. The following trade-offs can be observed from this figure. The harvested energy $e$ increases as $p$ increases. This is because higher $p$ corresponds to the symbols from the first layer $c=1$ which have higher energy (since $A_1 > A_2$) being used more frequently in $\mathscr{C}$. For a fixed value of $A_2$ in Fig.~\ref{FigBR}, the bound on the information rate $R$ first increases and then decreases as a function of $e$ in~\eqref{Eq33}. For each of these curves, the maximum $R = 2.13$ bits/channel use corresponds to the uniform type, \ie, $p = 0.5$. For $p$ lesser or greater than $0.5$, the bound on $R$ decreases. Furthermore, the bounds on $R$ are independent of the values of $A_1$ and $A_2$. This is due to the fact that the information rate $R$ is only a function of the number of codewords $M$ in~\eqref{Rfactorials}. Furthermore, the harvested energy $e$ also increases as $A_2$ increases. This is because, higher values of $A_2$ imply higher energy contained in the symbols in the second layer which in turn increases $e$. 

\begin{figure}[htb]
  \centering
  \includegraphics[width=0.7\textwidth]{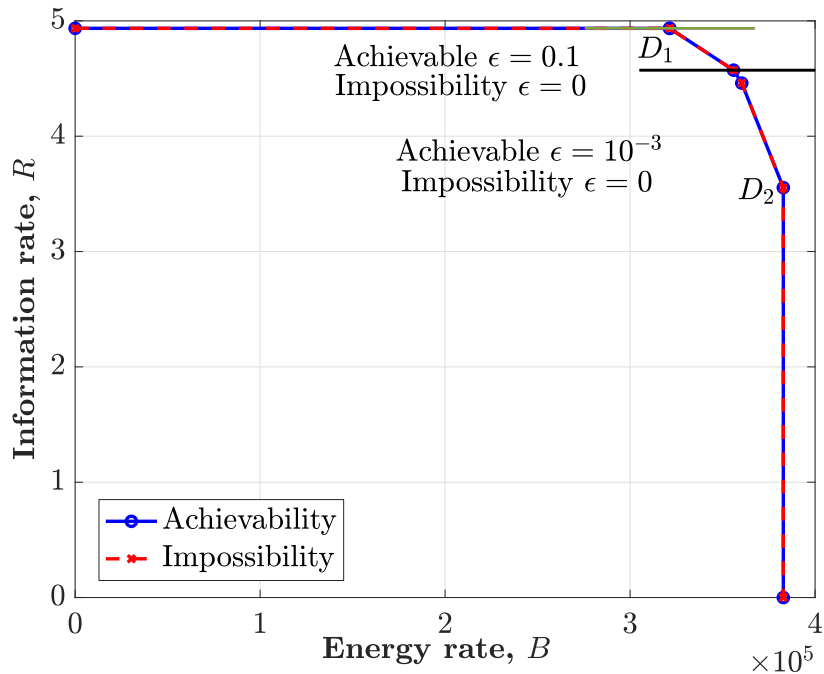}
  \caption{Impossibility and achievable information-energy regions for constant composition codes in the family ${\sf C} \left(C,\boldsymbol{A},\boldsymbol{L},\boldsymbol{\alpha},\boldsymbol{p},\boldsymbol{r} \right)$.}
\label{FigRegions}
\end{figure}

\subsubsection{Information-energy regions} \label{SecExamples}
Fig.~\ref{FigRegions} illustrates the impossibility and achievable information-energy regions of the constructed constant composition $(n,M,\mathcal{X},P,\epsilon,B,\delta)$-codes in the family ${\sf C} \left(C,\boldsymbol{A},\boldsymbol{L},\boldsymbol{\alpha},\boldsymbol{p},\boldsymbol{r} \right)$ in~\eqref{EqFamily}. 
For this example, consider constant composition $(n,M,\mathcal{X},P,\epsilon,B,\delta)$-codes in ${\sf C} \left(C,\boldsymbol{A},\boldsymbol{L},\boldsymbol{\alpha},\boldsymbol{p},\boldsymbol{r} \right)$ that employ the set of channel input symbols $\mathcal{X}$ of the form in~\eqref{EqConstellationCircle} with number of layers $C = 3$. The duration of the transmission in channel uses is $n = 80$. 
The radii of the decoding regions $r_c$ are assumed to be the same for all the layers i.e., for all $c \inCountK{C}$, the radius $r_c = r$ in~\eqref{EqDecodingCircle}. The value of $r$ is obtained according to~\eqref{EqRminEqual}. The amplitude of the first layer is $A_1 = 30$. Amplitudes of the second and third layers $A_2$ and $A_3$ are determined by $r$ according to~\eqref{EqAmplitudesDifference}. The points on the curves in Fig.~\ref{FigRegions} are obtained by varying $\epsilon$ and the probability vector $\boldsymbol{p}$ in~\eqref{Eqp}.

Fig.~\ref{FigRegions} shows the following trade-offs between the information and energy transmission rates in the impossibility and achievable curves.
Firstly, the maximum achievable information transmission rate is $R = 4.9$ bits/channel use. This $R$ is achieved by a code in which all the symbols in the set of channel inputs $\mathcal{X}$ are used with the same frequency. The maximum energy transmission rate that can be achieved at $R = 4.93$ bits/channel use is $B = 3.2 \times 10^5$ energy units. This corresponds to the point $D_1$ in Fig.~\ref{FigRegions}.
Secondly, the maximum achievable $B$ is $3.8 \times 10^{5}$ energy units. This is achieved by a code that exclusively uses the symbols in the first layer i.e., the probability vector $\boldsymbol{p}$ in~\eqref{Eqp} is $\boldsymbol{p} = (1,0,0)^{\sf{T}}$. The maximum $R$ that can be achieved at $B = 3.8 \times 10^{5}$ energy units is $R = 3.8$ bits/channel use. This corresponds to the point $D_2$ in Fig.~\ref{FigRegions}.
Thirdly, the curves between the points $D_1$ and $D_2$ in Fig.~\ref{FigRegions} show the trade-off between the information and energy transmission rates. As $B$ is increased from $3.2 \times 10^5$ energy units at point $D_1$, $R$ begins to decreases. Similarly, as $R$ is increased from $3.8$ bits/channel use at point $D_2$, $B$ begins to decrease.

The codes constructed in this work match the impossibility bounds in Section~\ref{sec:results} except for the DEP $\epsilon$. In Fig.~\ref{FigRegions}, the impossibility and achievable information-energy rate curves for $\mathscr{C}$ overlap. However, for the same information and energy rate pair, the DEP for the achievable curves is higher than the impossibility. The sub-optimality in DEP arises due to the choice of circular decoding regions in~\eqref{EqDecodingCircle} (See Remark~\ref{Remark1}). 
\subsubsection{Comparison with state of the art}
\begin{figure}[htb] 
  \centering
  \includegraphics[width=0.7\textwidth]{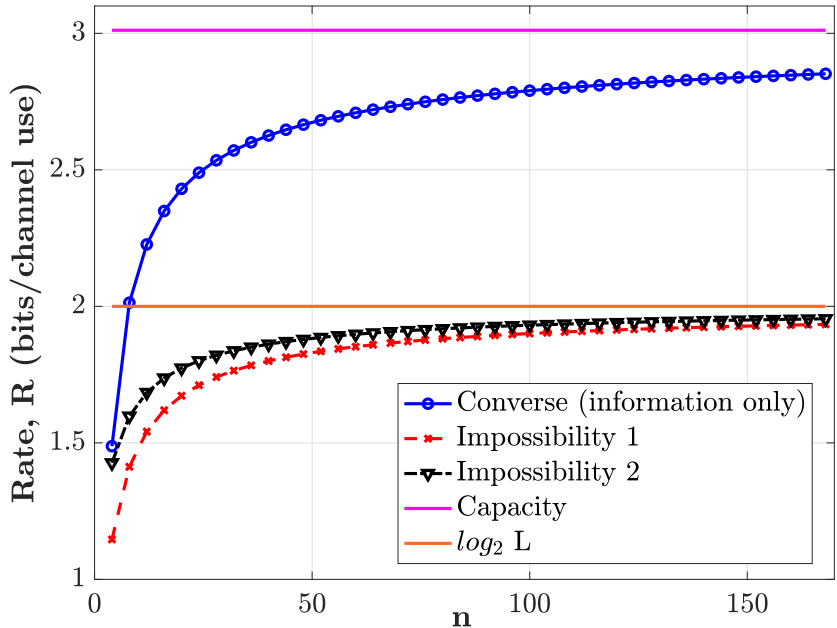}
  \caption{A comparison of the impossibility bounds on the information rate in~\eqref{EqRbound} (Impossibility 1) and~\eqref{EqCorRUpperBoundRelax} (Impossibility 2) for finite block-length SIET with the finite block-length converse in~\cite{polyanskiy2010channel} as a function of the block-length $n$.}
\label{FigCompare}
\end{figure}

Fig.~\ref{FigCompare} compares the impossibility bounds on the information rate proposed in~\eqref{EqRbound} (Impossibility 1) and~\eqref{EqCorRUpperBoundRelax} (Impossibility 2) with the finite block-length converse bound on the information rate in~\cite[Theorem 54]{polyanskiy2010channel} for a QPSK (\ie, $L=4$) constellation and signal to noise ratio (SNR) equal to $18$ dB. It should be noted that the converse bounds in~\cite{polyanskiy2010channel} are concerned with information transmission alone. 
The following observations can be made from these plots. Firstly, the proposed bounds are very close to the finite block-length converse for small block-lengths. As the block-length increases, the difference between the two bounds increases. Secondly, the impossibility bounds from~\eqref{EqRbound} and~\eqref{EqCorRUpperBoundRelax} are very close. This means that the approximation in~\eqref{EqCorRUpperBoundRelax} is tight, especially as the block-length increases. Finally, it should be noted that the bounds in~\eqref{EqRbound} and~\eqref{EqCorRUpperBoundRelax} are for the case of constant composition codes whereas the bound from~\cite{polyanskiy2010channel} has no such restriction. In fact, as stated in~\eqref{EqRbound}, the information rate $R$ can in fact not exceed the quantity $\log_2 L$. As can be observed from Fig.~\ref{FigCompare}, the proposed bounds actually come very close to the maximum value of $\log_2 L$ as the block-length $n$ increases.

 \bibliographystyle{IEEEtran}
\bibliography{myrefs}
 
 \appendix
 \subsection{Proof of Lemma~\ref{LemmaGeometry}} \label{AppA}
\begin{proof}
In Fig.~\ref{FigConstellation}, consider the circle of radius $r_c$ centered at $T$ and the larger circle of radius $A_c$ centered at the origin $O$. The circles intersect at points $S$ and $U$. The angle subtended by the major arc $\arc{\mbox{SU}}$ at $O$ is the reflex angle $2\pi - \beta$. Since the angle subtended by an arc of a circle at its centre is two times the angle that it subtends anywhere on the circumference, it holds that
\begin{equation}
    2\pi- \beta = 2\angle \mbox{STU},
\end{equation}
which implies that
\begin{equation} \label{EqHalfAngle}
    \angle \mbox{STU} = \frac{2\pi- \beta}{2}.
\end{equation}
The line segment $TO$ bisects angles $\angle \mbox{STU}$ and $\angle \mbox{SOU}$. Therefore, the following hold:
   \begin{IEEEeqnarray}{rCl}
        \angle \mbox{STO} &=& \frac{2\pi - \beta}{4}, \\
       \angle \mbox{SOT} &=& \frac{\beta}{2}.
   \end{IEEEeqnarray}
   From the triangle $\triangle \mbox{SOT}$, it holds that:
   \begin{IEEEeqnarray}{rCl}
  \frac{\sin \left(\angle \mbox{SOT} \right)}{ST} = \frac{\sin \left(\angle \mbox{STO} \right)}{SO}.
       \end{IEEEeqnarray}
       This implies that,
     \begin{IEEEeqnarray}{rCl}    
\frac{\sin \left(\frac{\beta}{2} \right) }{r_c} &=& \frac{\sin \left(\frac{2\pi-\beta}{4} \right)}{A_c}, \\
	&=& \frac{1}{A_c} \sin \left(\frac{\pi}{2} - \frac{\beta}{4} \right), \\
	&=&  \frac{1}{A_c} \cos \left( \frac{\beta}{4} \right). \label{Eq127a}
 \end{IEEEeqnarray}
From~\eqref{Eq127a}, it follows that,
       \begin{IEEEeqnarray}{rCl}    	
	  \frac{2}{r_c} \sin \left(\frac{\beta}{4} \right) \cos \left(\frac{\beta}{4} \right) = \frac{1}{A_c} \cos \left( \frac{\beta}{4} \right), 
	  \end{IEEEeqnarray}
	  which implies that,
	 \begin{IEEEeqnarray}{rCl}  
	 && \sin \left(\frac{\beta}{4} \right)  = \frac{r_c}{2A_c},  \quad \mbox{and}\\
       &&\beta = 4\arcsin{\frac{r_c}{2A_c}}.
    \end{IEEEeqnarray}
This completes the proof.
\end{proof}
\subsection{Proof of Lemma~\ref{LemmaAlphaImmaterial}} \label{AppB}
\begin{proof} 
From~\eqref{EqAvgGamma}, the average DEP for $\mathscr{C}$ is given by the following:
\begin{IEEEeqnarray}{rCl} 
    \gamma \left(\mathscr{C} \right) &=& 1 - \frac{1}{M} \sum_{i=1}^M \prod_{c=1}^C \prod_{\ell=1}^{L_c} \left(  \int_{\mathcal{G}_c^{(\ell)}} \frac{1}{\pi \sigma^2} \exp \left(- \frac{(\Re(y)-\Re(x_c^{(\ell)}))^2 +(\Im(y) - \Im(x_c^{(\ell)}))^2}{\sigma^2} \right) \mathrm{d}y  \right)^{n P_{\boldsymbol{u}(i)}(x_c^{(\ell)})},  \label{Eq129} \\
&=& 1 - \frac{1}{M} \sum_{i=1}^M \prod_{c=1}^C \prod_{\ell=1}^{L_c} \left( \int_{y: \left| y - x_c^{(\ell)} \right|^2 \leq r_c^2} \frac{1}{\pi \sigma^2}\exp \left( - \frac{\left| y - x_c^{(\ell)} \right|^2}{\sigma^2} \right) \mathrm{d}y   \right)^{n P_{\boldsymbol{u}(i)}(x_c^{(\ell)})},
\label{Eq130}
\end{IEEEeqnarray}
where, the expression in~\eqref{Eq130} follows from~\eqref{Eq129} due to~\eqref{Eq4a} and~\eqref{EqDecodingCircleComplex}.
The change of variable $y = z + x_c^{(\ell)} - \hat{x}_c^{(\ell)}$ in the integral in~\eqref{Eq130} yields
\begin{equation} \label{Eq127}
    \int_{y: \left| y - x_c^{(\ell)} \right|^2 \leq r_c^2} \frac{1}{\pi \sigma^2}\exp \left( - \frac{\left| y - x_c^{(\ell)} \right|^2}{\sigma^2} \right) \mathrm{d}y = \int_{z: \left| z - \hat{x}_c^{(\ell)} \right|^2 \leq r_c^2} \frac{1}{\pi \sigma^2}\exp \left( - \frac{\left| z - \hat{x}_c^{(\ell)} \right|^2}{\sigma^2} \right) \mathrm{d}z.
\end{equation}
Using~\eqref{Eq127} in~\eqref{Eq130} yields:
\begin{IEEEeqnarray}{rCl} \label{Eq131}
 \gamma \left(\mathscr{C} \right) &=&  1 - \frac{1}{M} \sum_{i=1}^M \prod_{c=1}^C \prod_{\ell=1}^{L_c} \left( \int_{z: \left| z - \hat{x}_c^{(\ell)} \right|^2 \leq r_c^2} \frac{1}{\pi \sigma^2}\exp \left( - \frac{\left| z - \hat{x}_c^{(\ell)} \right|^2}{\sigma^2} \right) \mathrm{d}z   \right)^{n P_{\boldsymbol{u}(i)}(x_c^{(\ell)})}.
\end{IEEEeqnarray}
The empirical probability of usage of symbols is equal for the symbols $x_c^{(\ell)}$ and $\hat{x}_c^{(\ell)}$. That is, for all $I \inCountK{M}$, all $c \inCountK{C}$ and all $\ell \inCountK{L_c}$, it holds that
\begin{equation} \label{EqEqualType}
    P_{\boldsymbol{u}(i)}(x_c^{(\ell)}) = P_{\boldsymbol{u}(i)}(\hat{x}_c^{(\ell)}).
\end{equation} 
From \eqref{Eq131} and~\eqref{EqEqualType}, it follows that
\begin{IEEEeqnarray}{rCl} 
 \gamma \left(\mathscr{C} \right) &=&  1 - \frac{1}{M} \sum_{i=1}^M \prod_{c=1}^C \prod_{\ell=1}^{L_c} \left( \int_{z: \left| z - \hat{x}_c^{(\ell)} \right|^2 \leq r_c^2} \frac{1}{\pi \sigma^2}\exp \left( - \frac{\left| z - \hat{x}_c^{(\ell)} \right|^2}{\sigma^2} \right) \mathrm{d}z   \right)^{n P_{\boldsymbol{u}(i)}(\hat{x}_c^{(\ell)})} \\
 &=& \gamma \left(\mathscr{C}' \right),
\end{IEEEeqnarray}
which is the desired result in~\eqref{Eq123}.

From~\eqref{EqShiftedSymbols}, for all $c \inCountK{C}$, all $\ell \inCountK{L_c}$, and $\omega \in [0,2\pi]$, it holds that,
\begin{equation} 
    \hat{x}_c^{\left( \ell \right)} = e^{\mathrm{i} \omega} x_c^{\left( \ell \right)}.
\end{equation} 
This implies that
\begin{equation} \label{Eq136}
  \left|\hat{x}_c^{\left( \ell \right)} \right| = \left| x_c^{\left( \ell \right)} \right|. 
\end{equation} 

Using~\eqref{Eq25} and~\eqref{Eq27c}, the EOP for $\mathscr{C}$ is given by
\begin{IEEEeqnarray}{rCl} 
 \theta(\mathscr{C},B) & = & \frac{1}{M} \sum_{i=1}^M \mathds{1}_{\left\lbrace k_1 \sum_{c=1}^C \sum_{\ell=1}^{L_c} n P_{\boldsymbol{u}(i)} \left( x_c^{(\ell)} \right) \left| x_c^{(\ell)} \right|^2 + k_2 \sum_{c=1}^C \sum_{\ell=1}^{L_c} n P_{\boldsymbol{u}(i)} \left( x_c^{(\ell)} \right) \left| x_c^{(\ell)} \right|^4 < B \right\rbrace} \\
 &=& \frac{1}{M} \sum_{i=1}^M \mathds{1}_{\left\lbrace k_1 \sum_{c=1}^C \sum_{\ell=1}^{L_c} n P_{\boldsymbol{u}(i)} \left( \hat{x}_c^{(\ell)} \right) \left| \hat{x}_c^{(\ell)} \right|^2 + k_2 \sum_{c=1}^C \sum_{\ell=1}^{L_c} n P_{\boldsymbol{u}(i)} \left( \hat{x}_c^{(\ell)} \right) \left| \hat{x}_c^{(\ell)} \right|^4 < B \right\rbrace} \label{Eq138} \\
&=& \theta \left(\mathscr{C}',B \right),
\end{IEEEeqnarray}
where, the equality in~\eqref{Eq138} follows from \eqref{EqEqualType} and~\eqref{Eq136}.
This completes the proof.
\end{proof}
\subsection{Proof of Lemma~\ref{LemmaAlphaImmaterialStrong}} \label{AppC}
\begin{proof} 
Denote the $\ell$\ts{th} symbol in layer $c$ of $\mathscr{C}$ by $x_c^{(\ell)}$ and that of $\mathscr{C}'$ by $\bar{x}_c^{(\ell)}$. From~\eqref{EqLayerCircle}, it follows that
\begin{IEEEeqnarray}{rCl}
x_c^{\left( \ell \right)} &=& A_c \exp\left(\mathrm{i} \left(\frac{2\pi}{L_c} \ell+\alpha_c\right)\right) \\
&=& A_c \exp\left(\mathrm{i} \left(\frac{2\pi}{L_c} \ell+\alpha_c + \alpha_c' - \alpha_c'\right)\right) \\
&=& A_c \exp\left(\mathrm{i} \left(\frac{2\pi}{L_c} \ell+\alpha_c' \right)\right) \exp\left(\mathrm{i} \left(\alpha_c - \alpha_c'\right)\right) \\
&=& \exp\left(\mathrm{i} \left(\alpha_c - \alpha_c'\right)\right)  \bar{x}_c^{\left( \ell \right)} \\
&=& \exp\left(\mathrm{i} \omega_c \right)  \bar{x}_c^{\left( \ell \right)},
\end{IEEEeqnarray}
where, $\omega_c = \alpha_c - \alpha_c' \in [0,2 \pi]$.
This implies that for all $c \inCountK{C}$ and all $\ell \inCountK{L_c}$, it holds that
\begin{IEEEeqnarray}{rCl} \label{Eq194}
\left| x_c^{\left( \ell \right)} \right| &=& \left|  \bar{x}_c^{\left( \ell \right)} \right|.
\end{IEEEeqnarray}
The result in~\eqref{Eq173} then follows from~\eqref{Eq194} and Lemma~\ref{LemmaAlphaImmaterial}.

From~\eqref{Eq130}, the average DEP for $\mathscr{C}$ is given by
\begin{IEEEeqnarray}{rCl} 
    \gamma \left(\mathscr{C} \right) &=& 1 - \frac{1}{M} \sum_{i=1}^M \prod_{c=1}^C \prod_{\ell=1}^{L_c} \left( \int_{y: \left| y - x_c^{(\ell)} \right|^2 \leq r_c^2} \frac{1}{\pi \sigma^2}\exp \left( - \frac{\left| y - x_c^{(\ell)} \right|^2}{\sigma^2} \right) \mathrm{d}y   \right)^{n P_{\boldsymbol{u}(i)}(x_c^{(\ell)})}    \label{Eq165} \\
    &=& 1 - \frac{1}{M} \sum_{i=1}^M \prod_{c=1}^C \prod_{\ell=1}^{L_c} \left( \int_{y: \left| y - x_c^{(\ell)} \right|^2 \leq r_c^2} \frac{1}{\pi \sigma^2}\exp \left( - \frac{\left| y - x_c^{(\ell)} \right|^2}{\sigma^2} \right) \mathrm{d}y   \right)^{n P_{\boldsymbol{u}(i)}(\bar{x}_c^{(\ell)})},
    \label{Eq166}
\end{IEEEeqnarray}
where, the equality in~\eqref{Eq166} follows from~\eqref{Eq165} because $P_{\boldsymbol{u}(i)}(x_c^{(\ell)}) = P_{\boldsymbol{u}(i)}(\bar{x}_c^{(\ell)})$.
Using the change of variable $y = z + x_c^{(\ell)} - \bar{x}_c^{(\ell)}$ in~\eqref{Eq166} yields
\begin{IEEEeqnarray}{rCl} 
    \gamma \left(\mathscr{C} \right) &=& 1 - \frac{1}{M} \sum_{i=1}^M \prod_{c=1}^C \prod_{\ell=1}^{L_c} \left( \int_{z: \left| z - \bar{x}_c^{(\ell)} \right|^2 \leq r_c^2} \frac{1}{\pi \sigma^2}\exp \left( - \frac{\left| z - \bar{x}_c^{(\ell)} \right|^2}{\sigma^2} \right) \mathrm{d}z   \right)^{n P_{\boldsymbol{u}(i)}(\bar{x}_c^{(\ell)})} \\
    &=& \gamma \left(\mathscr{C}' \right),
\end{IEEEeqnarray}
which completes the proof of~\eqref{Eq159}.
\end{proof}
\subsection{Proof of Lemma~\ref{LemmaRateMax}} \label{AppD}
\begin{proof} 
For the $(n,M,\mathcal{X},P,\epsilon,B,\delta)$-code $\mathscr{C}'$ in the family ${\sf C} \left(C,\boldsymbol{A},\boldsymbol{L},\boldsymbol{\alpha},\boldsymbol{p},\boldsymbol{r} \right)$ with $p_c$ of the form in~\eqref{EqpcRmax}, from~\eqref{EqTypeCircle}, for all $c \inCountK{C}$ and all $\ell \inCountK{L_c}$, the type $P_{\mathscr{C}'}$ is given by the following:
\begin{IEEEeqnarray}{rCl}
    P_{\mathscr{C}'}(x_c^{(\ell)}) = \frac{p_c}{L_c} = \frac{1}{L}. \label{Eq210}
\end{IEEEeqnarray}

  For the $(n,M,\mathcal{X},P,\epsilon,B,\delta)$-code $\mathscr{C}'$ with the type $P_{\mathscr{C}'}$  of the form in~\eqref{Eq210}, the number of codewords that can be represented using $L$ symbols is given by $M = L^n$. From~\eqref{EqR}, the information transmission rate for code $\mathscr{C}'$ is given by
    \begin{IEEEeqnarray}{rCl}
    R \left(\mathscr{C}' \right) = \frac{\log_2 M}{n} = \log_2 L.
\end{IEEEeqnarray}
  For the $(n,M,\mathcal{X},P,\epsilon,B,\delta)$-code $\mathscr{C}$, from~\eqref{EqR} and Lemma~\ref{lemma:R_upperbound}, it follows that,
  \begin{IEEEeqnarray}{rCl}
    R \left(\mathscr{C} \right) \leq \log_2 L = R \left(\mathscr{C}' \right).
\end{IEEEeqnarray}
This completes the proof.
\end{proof}
\subsection{Proof of Lemma~\ref{LemmaMaxEnergy}} \label{AppE}
\begin{proof} 
Since $p_c = 0$ for all $c \in \left \lbrace 2,3, \ldots, C \right \rbrace$, from~\eqref{Eq115i} it follows that, for all $i \inCountK{M}$, all $c \in \left \lbrace 2,3, \ldots, C \right \rbrace$ and all $\ell \inCountK{L_c}$, the type 
\begin{IEEEeqnarray}{rCl} \label{Eq2344}
P_{\boldsymbol{u}(i)} \left( x_c^{(\ell)} \right) = 0.
\end{IEEEeqnarray}
Since $P_{\boldsymbol{u}(i)}$ is a pmf as defined in~\eqref{eq:u_measure}, it holds that, for all $i \inCountK{M}$,
\begin{IEEEeqnarray}{rCl} \label{Eq2355}
\sum_{c=1}^C \sum_{\ell = 1}^{L_c} P_{\boldsymbol{u}(i)} \left( x_c^{(\ell)} \right) = 1.
\end{IEEEeqnarray}
From~\eqref{Eq2344} and~\eqref{Eq2355}, it follows that,
\begin{IEEEeqnarray}{rCl} \label{Eq236}
\sum_{\ell = 1}^{L_1} P_{\boldsymbol{u}(i)} \left( x_1^{(\ell)} \right) = 1.
\end{IEEEeqnarray}

From~\eqref{Eq25} and~\eqref{Eq27c}, the EOP for the code $\mathscr{C}$ is given by
\begin{IEEEeqnarray}{rCl} 
 \theta(\mathscr{C},B) & = & \frac{1}{M} \sum_{i=1}^M \mathds{1}_{\left\lbrace k_1 \sum_{c=1}^C \sum_{\ell=1}^{L_c} n P_{\boldsymbol{u}(i)} \left( x_c^{(\ell)} \right) \left| x_c^{(\ell)} \right|^2 + k_2 \sum_{c=1}^C \sum_{\ell=1}^{L_c} n P_{\boldsymbol{u}(i)} \left( x_c^{(\ell)} \right) \left| x_c^{(\ell)} \right|^4 < B \right\rbrace} \\
 &=& \frac{1}{M} \sum_{i=1}^M \mathds{1}_{\left\lbrace k_1 \sum_{\ell=1}^{L_1} n P_{\boldsymbol{u}(i)} \left( x_1^{(\ell)} \right) A_1^2 + k_2  \sum_{\ell=1}^{L_1} n P_{\boldsymbol{u}(i)} \left( x_1^{(\ell)} \right) A_1^4 < B \right\rbrace} \label{Eq238} \\
 & = & \frac{1}{M} \sum_{i=1}^M \mathds{1}_{\left\lbrace k_1 n A_1^2 + k_2 n A_1^4 < B \right\rbrace}, \label{Eq219}
\end{IEEEeqnarray}
where, the equality in~\eqref{Eq238} follows from~\eqref{Eq2344} and~\eqref{Eq219} follows from~\eqref{Eq236}.
Similarly, the EOP for the code $\mathscr{C}'$ is given by
\begin{IEEEeqnarray}{rCl} 
 \theta(\mathscr{C}',B) & = & \frac{1}{M} \sum_{i=1}^M \mathds{1}_{\left\lbrace k_1 \sum_{c=1}^C \sum_{\ell=1}^{L_c} n P_{\boldsymbol{u}(i)} \left( x_c^{(\ell)} \right) \left| x_c^{(\ell)} \right|^2 + k_2 \sum_{c=1}^C \sum_{\ell=1}^{L_c} n P_{\boldsymbol{u}(i)} \left( x_c^{(\ell)} \right) \left| x_c^{(\ell)} \right|^4 < B \right\rbrace} \\
 &=& \frac{1}{M} \sum_{i=1}^M \mathds{1}_{\left\lbrace k_1 \sum_{c=1}^C \sum_{\ell=1}^{L_c} n P_{\boldsymbol{u}(i)} \left( x_c^{(\ell)} \right) A_c^2 + k_2 \sum_{c=1}^C \sum_{\ell=1}^{L_c} n P_{\boldsymbol{u}(i)} \left( x_c^{(\ell)} \right) A_c^4 < B \right\rbrace} \label{Eq220} \\
 &\geq& \frac{1}{M} \sum_{i=1}^M \mathds{1}_{\left\lbrace k_1 \sum_{c=1}^C \sum_{\ell=1}^{L_c} n P_{\boldsymbol{u}(i)} \left( x_c^{(\ell)} \right) A_1^2 + k_2 \sum_{c=1}^C \sum_{\ell=1}^{L_c} n P_{\boldsymbol{u}(i)} \left( x_c^{(\ell)} \right) A_1^4 < B \right\rbrace} \label{Eq221} \\
 & = & \frac{1}{M} \sum_{i=1}^M \mathds{1}_{\left\lbrace k_1 n A_1^2 + k_2 n A_1^4 < B \right\rbrace} \label{Eq222} \\
 &=& \theta(\mathscr{C},B), \label{Eq223}
\end{IEEEeqnarray}
where, the equality in~\eqref{Eq220} follows from~\eqref{Eq218}; the inequality in~\eqref{Eq221} follows from~\eqref{EqAmplitudesOrder}; the equality in~\eqref{Eq222} follows from~\eqref{Eq2355}; and~\eqref{Eq223} follows from~\eqref{Eq219}.
This completes the proof.
\end{proof}

%
%
%
%
%
%
%
%

\end{document}